\documentclass[runningheads]{llncs}

\usepackage{tikz}

\newcommand {\imp} {\rightarrow}

\newcommand {\sx} {\langle}
\newcommand {\dx} {\rangle}

\newcommand {\emme} {\mathcal{M}}

\newcommand {\tc} {\mid}
\newcommand {\vuoto} {\emptyset}

\newcommand{\U}{{\cal U}}

\newcommand{\alc}{\mathcal{ALC}}
\newcommand{\al}{\mathcal{LC}}

\newcommand{\alcio}{\mathcal{ALCOI}}
\newcommand{\shiq}{\mathcal{SHIQ}}
\newcommand{\shoiq}{\mathcal{SHOIQ}}
\newcommand{\sroiq}{\mathcal{SROIQ}}

\newcommand{\be}{\begin{enumerate}}
\newcommand{\ee}{\end{enumerate}}

\newcommand{\hide}[1]{}

\newcommand{\HF}{\mathsf{HF}}

\def \cases{\left \{\begin{array}{l}}
\def \endcases{\end{array}\right .}

\newcommand {\bes} {\begin{description}}
\newcommand{\ens} {\end{description}}

\newcommand {\beq} {\begin{quote}}
\newcommand {\enq} {\end{quote}}
\newcommand {\bit} {\begin{itemize}}
\newcommand {\enit} {\end{itemize}}

\newcommand {\window}{ 
	\begin{tikzpicture}
	\draw rectangle (0.2,0.2);
	\draw rectangle (0.1,0.2);
	\end{tikzpicture} \  }

\newcommand{\nats}{\mathbb{N}}

\newenvironment{pozz}{\color{black}}{\color{black}}

\usepackage{pgf}
\usepackage{tikz}
\usetikzlibrary{arrows,automata}

\usepackage{times}
\usepackage{undertilde}
\usepackage{graphicx}
\usepackage{latexsym}
\usepackage{graphicx}
\usepackage{color}
\usepackage{amssymb}
\usepackage{colortbl}
\usepackage{amsmath}
\usepackage{microtype}

\usepackage[ddmmyy,hhmmss]{datetime}

\title{Adding the Power-Set to Description Logics}

\author{Laura Giordano \inst{1} \and Alberto Policriti \inst{2}}
\institute{DISIT - Universit\`a del Piemonte Orientale - 
 Alessandria, Italy - \email{\small laura.giordano@uniupo.it} \and
 Dipartimento di Scienze Matematiche, Informatiche e Fisiche, Universit\`a di Udine\\
Istituto di Genomica Applicata, Parco Scientifico e Tecnologico ``L. Danieli",  Italy - \\ \email{\small alberto.policriti@uniud.it}}

\begin{document}

 \maketitle


\begin{abstract} 
We explore the relationships between Description Logics and  Set Theory. The study is carried on using, on the set-theoretic side, a very rudimentary axiomatic set theory $ \Omega $, consisting of only four axioms characterizing binary union, set difference, inclusion, and the power-set. 
An extension of $\alc$, $\alc^\Omega$,  is then defined 
in which concepts are naturally interpreted as sets living in $\Omega$-models.
In $\alc^\Omega$ not only membership between concepts is allowed---even admitting circularity---but also the power-set construct is exploited 
to add metamodeling capabilities. 
We investigate translations of $\alc^\Omega$ into standard description logics as  well as a set-theoretic translation.
A polynomial encoding of  $\alc^\Omega$ in $\alcio$ proves the validity of the finite model property as well as 
an \textsc{ExpTime} upper bound on the complexity of concept satisfiability. 
We develop a set-theoretic translation of $\alc^\Omega$ in the theory $\Omega$, exploiting a technique 
proposed 
for translating normal modal and polymodal logics into $\Omega$. Finally, we show that the fragment $ \al^{\Omega} $ of $\alc^\Omega$, which does not admit roles and individual names, is as expressive as $\alc^\Omega$
\end{abstract}

 \section{Introduction}
 
 Concept and concept constructors in Description Logics (DLs)  allow  
 to manage information built-up and stored as collections of elements of a given domain.
In this paper we would like to take the above statement ``seriously'' and put forward a DL
 doubly linked with a (very simple, axiomatic) set theory. Such a logic will be suitable to manipulate concepts (also called classes in OWL  \cite{OWL}) as first-class citizens, in the sense that it will allow the possibility to have concepts as instances (a.k.a. elements) of other concepts. From the set-theoretic point of view this is \emph{the} way to proceed, as stated in the following quotation from the celebrated \emph{Naive Set Theory} (\cite{Halmos}):
 \begin{quote}
\emph{Sets, as they are usually conceived, have {elements} or {members}. An element of a set may be a wolf, a grape, or a pidgeon. It is important to know that a set itself may also be an element of some other set. [...] What may be surprising is not so much that sets may occur as elements, but that for mathematical purposes no other elements need ever be considered.}
 
 \hfill {\sc P. Halmos}
 \end{quote}

Also in the Description Logic arena the idea of enhancing the language of description logics with statements of the form $ C \in D $, with $C$ and $D$ concepts, is not new, as assertions of the form $D(A)$, with $A$ a concept name, are already allowed  in OWL-Full \cite{OWL}.
Here,  we do not consider roles, i.e. relations among individuals (also called properties in OWL),  as possible instances of concepts. However,
we would like to push the usage of membership among concepts a little forward, allowing not only the possibility of stating that an arbitrary concept $C$ can be thought of as an instance of another one ($ C \in D $)---or even as an instance of itself ($ C \in C $)---but also opening to the possibility of talking about \emph{all possible} sub-concepts of $ C $, that is adding memberships to the  \textit{power-set}  $ \texttt{Pow}(C) $ of $ C $.

In order to realize our plan we introduce a DL, to be dubbed $ \alc^{\Omega} $, whose two parents are $ \alc $ and a rudimentary (finitely axiomatized) set theory $ \Omega $.

Considering an example taken from \cite{Welty1994,Motik05}, 
using membership axioms, we can represent the fact that eagles are { in the red list of endangered species}
by the axiom $\mathit{Eagle \in RedListSpecies}$ and that Harry is an eagle, by the assertion $\mathit{Eagle(harry)}$.  
We could further consider a concept $ \mathit{ModifiableList}$, consisting of those lists that can be modified 
and, for example, it would be reasonable to ask $\mathit{RedListSpecies \in  Modifiable}$ $\mathit{List }$ but, more interestingly, we would also clearly want  $\mathit{ModifiableList \in ModifiableList}$. 

The power-set concept, $ \texttt{Pow}(C) $, 
allows to capture in a natural way the interactions between concepts and metaconcepts.
Considering again the example above,
the statement ``all  instances of species in the red list are not allowed to be hunted",
can be represented by the  concept inclusion axiom:
$\mathit{RedListSpecies \sqsubseteq \texttt{Pow}(CannotHunt)}$, meaning that
all the instances in the $\mathit{RedListSpecies}$ (as the class $\mathit{Eagle}$) are collections of individuals of the class $\mathit{CannotHunt}$. 
Notice, however, that $\texttt{Pow}(CannotHunt) $ is not limited to include $\mathit{RedListSpecies}$ but can include a much larger universe of sets (e.g. anything belonging to  $ \texttt{Pow}(\mathit{Humans}) $).

Motik has shown in \cite{Motik05} that the semantics of metamodeling adopted in OWL-Full leads to undecidability already for $\alc$-Full, due to the free mixing of logical and metalogical symbols. In \cite{Motik05}, limiting this free mixing but allowing atomic names to be interpreted as concepts and to occur as instances of other concepts,  two alternative semantics (the Contextual $\pi$-semantics and the Hilog $\nu$-semantics) are proposed for metamodeling.  Decidability of $\shoiq$ extended with metamodeling is proved under either one of the two proposed semantics. 
Many other approaches to metamodeling have been proposed in the literature, including membership among concepts. 
Most of them  \cite{Badea1997,DeGiacomo2011,HomolaDL14,KubincovaDL15,Gu2016} 
are based on a Hilog semantics, while \cite{Pan2005,Motz2015} define extensions of OWL DL and of $\shiq$ (respectively), 
based on semantics interpreting concepts as well-founded sets---i.e. sets with no cycles or infinite descending chains of $ \in $-related sets.  
None of these proposals includes the power-set concept constructor in the language
apart from \cite{Badea1997}, where a way of representing the power-set 
in description logics was suggested.

Here, we propose an extension of $\alc$ with power-set concepts and membership axioms among concepts, whose semantics is naturally defined using sets (not necessarily well-founded) living in $\Omega$-models. 
We first prove that $ \alc^{\Omega} $ is decidable by defining, for any $ \alc^{\Omega} $ knowledge base $K$, a polynomial translation $K^T$ into $\alcio$,
exploiting a technique---originally proposed and studied in \cite{DAgostino1995} for defining a set-theoretic translation of normal modal logics ---consisting in identifying the membership relation $ \in $ with the accessibility relation of a normal modality.  
Such an identification naturally leads to a correspondence between the power-set operator and the modal necessity operator $ \Box $, a correspondence used here to translate power-set concepts into $ \forall R.C $-type concepts.
We show that the translation $K^T$ enjoys the finite model property and exploit it in the proof of completeness of the translation. 
From the translation in $\alcio$ we also get an \textsc{ExpTime} upper bound on the complexity of concept satisfiability in $ \alc^{\Omega} $.
Interestingly enough, our translation has strong relations with the first-order reductions in \cite{GlimmISWC2010,HomolaDL14,KubincovaDL15}.

We further exploit 
the correspondence between $\in$ and the accessibility relation of a normal modality
in another direction (the direction considered in  \cite{DAgostino1995}), to provide a polynomial set-theoretic translation of $\alc^\Omega$ in the set theory $\Omega$.
Our aim is to understand the real nature of the power-set concept in $\alc^\Omega$, as well as showing that a description logic with just the power-set concept, but no roles and no individual names, is as expressive as $\alc^\Omega$.
 
We proceed step by step by first defining a set-theoretic translation of $\alc$  with empty ABox  (in Section 5.1), directly exploiting Schild's correspondence result \cite{Schild91} and 
the set-theoretic translation for normal polymodal logics in  \cite{DAgostino1995}.
Then, we extend the translation to $\alc^\Omega$, first considering (in Section 5.2) the fragment of $\alc^\Omega$ containing union, intersection, (set-)difference, complement, and power-set (but neither roles nor named individuals) and we  show that this fragment,  that we call $\al^\Omega$, has an immediate set-theoretic translation into $\Omega$, where the power-set concept is translated to the power-set in $\Omega$.
Finally, (in Section 5.3) we provide an encoding of the whole $\alc^\Omega$ into $\al^\Omega$. This encoding shows that $\al^\Omega$ is as expressive as $\alc^\Omega$ and also provides, as a by-product, a set-theoretic translation of $\alc^\Omega$
where the membership relation $\in$ is used to capture both the roles $R_i$ 
and the membership relation in $\alc^\Omega$. The full path leads to a set-theoretic translation of both
 the universal restriction and power-set concept  of $\alc^\Omega$ in the theory $\Omega$ using the single relational symbol $ \in $.
 
 The outline of the paper is the following. Section 2 recalls the definition of the description logics $\alc$ and $\alcio$, and of the set theory $\Omega$.
  Section 3  introduces the logic $\alc^\Omega$. 
  Section 4  provides a translation of the logic $\alc^\Omega$ into the description logic $\alcio$.
  Section 5 develops set-theoretic translations for $\alc$ and $\al^\Omega$ and an encoding of $\alc^\Omega$ into $\al^\Omega$.
 Section 6 contains some discussion, Section 7 describes related work and Section 8 concludes the paper.

\section{Preliminaries}

\subsection{The description logics $\alc$ and $\alcio$}\label{sec:ALC}

Let ${N_C}$ be a set of concept names, ${N_R}$ a set of role names
  and ${N_I}$ a set of individual names.  
The set  ${\cal C}$ of $\alc$ \emph{concepts} can be
defined inductively as follows:

-  $A \in N_C$, $\top$ and $\bot$ are {concepts} in ${\cal C}$;
    
-  if $C, D \in {\cal C}$ and $R \in N_R$, then $C
\sqcap D, C \sqcup D, \neg C, \forall R.C, \exists R.C$ are
{concepts} in ${\cal C}$.

\noindent  
A knowledge base (KB) $K$ is a pair $({\cal T}, {\cal A})$, where ${\cal T}$ is a TBox and
${\cal A}$ an ABox.
${\cal T}$ is  a set of concept inclusions (or subsumptions) $C \sqsubseteq D$, where $C,D$ are concepts in ${\cal C}$.
${\cal A}$ is  a set of assertions of the form $C(a)$ 
and $R(a,b)$ where $C$ is a  concept, $R \in
N_R$, and $a, b \in N_I$.

An {\em interpretation} for $\alc$ (see \cite{handbook}) is a pair $I=\langle \Delta, \cdot^I \rangle$ where:
$\Delta$ is a domain---a set whose elements are denoted by $x, y, z, \dots$---and 
$\cdot^I$ is an extension function that maps each
concept name $C\in N_C$ to a set $C^I \subseteq  \Delta$, each role name $R \in N_R$
to  a binary relation $R^I \subseteq  \Delta \times  \Delta$,
and each individual name $a\in N_I$ to an element $a^I \in  \Delta$.
It is extended to complex concepts  as follows:
$\top^I=\Delta$, $\bot^I=\vuoto$, $(\neg C)^I=\Delta \backslash C^I$, $(C \sqcap D)^I =C^I \cap D^I$, $(C \sqcup D)^I=C^I \cup D^I$, and	
\begin{align*}
	&(\forall R.C)^I =\{x \in \Delta \tc \forall y. (x,y) \in R^I \imp y \in C^I\} \\ 
	&(\exists R.C)^I =\{x \in \Delta \tc \exists y.(x,y) \in R^I \ \& \ y \in C^I\}.
\end{align*}
%
%
%
%
The notion of satisfiability of a KB  in an interpretation is defined as follows:
\begin{definition}[Satisfiability and entailment] \label{satisfiability}
Given an $\alc$ interpretation $I=\langle \Delta, \cdot^I \rangle$: 

-	 $I$  satisfies an inclusion $C \sqsubseteq D$ if   $C^I \subseteq D^I$;

-	   $I$ satisfies an assertion $C(a)$ if $a^I \in C^I$ and an assertion $R(a,b)$ if $(a^I,b^I) \in R^I$. 

\noindent
 Given  a KB $K=({\cal T}, {\cal A})$, 
 an interpretation $I$  satisfies ${\cal T}$ (resp. ${\cal A}$) if $I$ satisfies all  inclusions in ${\cal T}$ (resp. all assertions in ${\cal A}$);
 $I$ is a \emph{model} of $K$ if $I$ satisfies ${\cal T}$ and ${\cal A}$.

 Let a {\em query} $F$ be either an inclusion $C \sqsubseteq D$ (where $C$ and $D$ are concepts) 
or an assertion $C(a)$:
 {\em $F$ is entailed by $K$}, written $K \models F$, if for all models $I=$$\sx \Delta,  \cdot^I\dx$ of $K$,
$I$ satisfies $F$.
\end{definition}
Given a knowledge base $K$,
the {\em subsumption} problem is the problem of deciding whether a given  inclusion $C \sqsubseteq D$ is entailed by  $K$.
The {\em instance checking} problem is the problem of deciding whether a given assertion $C(a)$ is entailed by $K$.
The {\em concept satisfiability} problem w.r.t. a knowledge base $K$ is the problem of deciding, for a given concept $C$, whether $C$ is \emph{consistent} with $K$ 
(i.e., whether there exists a model $I$ of $K$, such that $C^I \neq \emptyset$).

\medskip

In the following we will also consider the description logic $\alcio$ allowing inverse roles and nominals.
For a role $R \in N_R$, its {\em inverse} is a role, denoted by $R^-$, which can be used in existential and universal restrictions
with the following semantics:
$(x,y) \in (R^-)^I $ {\em if and only if}  $(y,x) \in R^I.$
For a named individual $a \in N_I$, the {\em nominal} $\{a\}$ is the concept such that:
$(\{a\})^I= \{a^I\}$.

\subsection{The theory $\Omega$}

The first-order axiomatic set theory $\Omega$ at the ground of our translation, consists of four extremely simple axioms (partially) characterizing the binary constructors \emph{union} and \emph{set-difference}, as well as the \emph{power-set} constructor. The underlying language is reduced to the relation symbols denoting \emph{membership} and \emph{inclusion}. More formally: 

\begin{definition} Consider a first order language with two binary relational symbols denoted by $ \in $ and $ \subseteq $, (to be used in infix notation). Let  $\cup$ and $\backslash$ two binary functional symbols (also  used in the customary infix notation) and let $\mathit{Pow}$ be a unary function symbol. 

The axiomatic set theory $ \Omega $ consists of the following collection of four axioms: 
\begin{enumerate}
\item $x \in y \cup  z   \leftrightarrow  x \in y \vee x \in z$; 
\item $x \in y \backslash z  \leftrightarrow  x \in y \wedge  x \not \in z$;
\item $x \subseteq y   \leftrightarrow  \forall z (z \in x \rightarrow z \in y)$; 
\item $x \in \mathit{Pow}(y)  \leftrightarrow  x \subseteq y$,
\end{enumerate}
completed with the standard deduction rules of \emph{generalization} and \emph{modus-ponens}.
\end{definition}
The above theory must be intended as a minimal Hilbert-style axiomatic system for set theory.
When thinking of specific models of $ \Omega $, however, we can clearly think of structures satisfying extra axioms. In particular, for example, the familiar \emph{well-founded} models of set theories, are perfectly legitimate models of $ \Omega $, in which the extra axiom of well-foundedness---implying that $ \in $ cannot form cycles or infinite descending chains---holds.
For instance let $x=\{\emptyset, \{\emptyset\} \}$.
$x$ is a finite well-founded set, and the sets $\emptyset$ and $\{\emptyset\}$ are elements of $x$. 
Instead, the set $y= \{\emptyset,\{\emptyset,\{\emptyset,\{ \ldots\} \} \}$ is finite but not well-founded. 

Whatever the axioms satisfied by the $\Omega$-model under consideration are, however,   \emph{everything} in the domain of such a model is supposed to be a \emph{set}. As a consequence, a set will have (only) sets as its elements. 
Moreover, as observed, circular definitions of sets are not forbidden. That is, for example, there are models of $ \Omega $ in which there are sets admitting themselves as elements. 
 For instance, the set $y$ above could simply be defined as  $y= \{\emptyset,y\}$ and has elements $\emptyset$ and $y$  itself. 

Finally, not postulating in $ \Omega $ any explicit  ``axiomatic link'' between membership $ \in $ and equality---more precisely: having no \textit{extensionality} axiom---, there exist $ \Omega $-models in which there are different sets with equal collection of elements. 
One (elementary) consequence of the extensionality axiom is the familiar fact that if $ a\subseteq b $ and $ b\subseteq a $, then $ a=b $. In non-extensional models, 
instead, there can be pairwise distinct sets included in each other.
The set $x'=\{a,b, \{b,c\} \}$ with $a,b$ and $c$ pairwise distinct and such that $a,b,c \subseteq \emptyset$, does not satisfy extensionality as $a$, $b$ and $c$ are different sets with the same (empty) extension.

\begin{definition}
$\Omega$-models are first order interpretations $\emme=({\cal U}, \cdot^\emme)$ satisfying the axioms of the theory $\Omega$.
The universe  ${\cal U}$ is the domain of interpretation of $\emme$ and,  $\cdot^\emme$ is an interpretation function mapping  each  symbol $\mathit{Pow}$, $\cup$, $\backslash$ of the language to a function over ${\cal U} $  (that is, $\mathit{Pow^\emme}: {\cal U} \rightarrow {\cal U}$,  \ 
$\cup^\emme: {\cal U} \times {\cal U} \rightarrow {\cal U}$ and  $\backslash^\emme: {\cal U} \times {\cal U} \rightarrow {\cal U}$)
and each predicate symbol $\in$ and $\subseteq$ to a binary relation over ${\cal U} $ (that is, $\in^\emme$ and $\subseteq^\emme$).
\end{definition}
Below, for sake of readability, we will avoid superscripts in  $\in^\emme$, $\subseteq^\emme$, $\mathit{Pow^\emme}$, $\cup^\emme$, $\backslash^\emme$.

\medskip 
Observe that the universe ${\cal U}$ of any $\Omega$-model $\emme$ must be infinite, as any element in ${\cal U}$ must have its power-set in ${\cal U}$
and, 
as an elementary consequence of Cantor's Theorem (see \cite{Jech:2003ly}), $|Pow(x)| > |x|$ when $|x|$ is finite. 
This closure with respect to the use of the power-set operator produces the most natural $ \Omega $-model, a well-founded one in which extensionality holds---and hence different sets are, in fact, extensionally different. 
\begin{definition}
The \emph{hereditarily finite well-founded sets} $ \HF $ denote the $ \Omega $-model $\emme=({\cal U}, \cdot^\emme)$ such that

- \ $ \cal U $ is  $ \HF=\bigcup_{n\in \nats} \HF_n$, where $\HF_0  =\emptyset;$ and $\HF_{n+1}  =\mathit{Pow}(\HF_n)$, for all $ i \in \nats $;

-  \  $ \cdot^{\emme} $ it the natural interpretations of $ \in $, $ \subseteq $, $\mathit{Pow}$, $\cup$ and $\backslash$ in $ \HF $.
\end{definition}


\noindent
By the above observation  $HF$ is minimal among the well-founded models of $\Omega$ in that it can be {\em embedded}  in any model of $\Omega$.
In $ \HF $ (sometimes denoted also as $ \HF^0 $) every system of set-theoretic equations of the form:
\[\left\{
\begin{array}{ccc}
x_1 & = &  \{x_{1,1}, \ldots , x_{1,m_1}\}; \\
x_2 & = & \{x_{2,1}, \ldots , x_{2,m_2}\}; \\
\vdots &  & \vdots \\
x_n & = &  \{x_{n,1}, \ldots , x_{n,m_n}\}, \\
\end{array}
\right. \]
where $n$ is finite, and $ x_{i,j}$  is one among $ x_1, \ldots, x_{i-1}$  for  $i = 1,\ldots,n$ and $ j= 1, \ldots, m_i $, finds a unique  solution. 
 Hence, $ \HF $  can be even identified with the collection of such systems of equations which, taken individually, are actually in bijective correspondence with the adjacency matrices of finite graphs. 

Insisting that $ x_{i,j}$ must be one of the left-hand side of equations defining an $ x_{ k}$ with $ k <i $, guarantees that a solution can be found in an ordered manner. In fact, it can be easily proved (even by an elementary graph-theoretic argument), that whenever a solution exists, every $ x_{i} $ can be found in $ \HF_{i+1} $.
As an example, the set $x=\{\emptyset, \{\emptyset\} \}$ above is in $\HF^0$, and can be defined by the system of equations: 
$x_{1}=\{x_{2},x_{3}\}; x_{2}= \{x_{3}\}; x_{3}=\{ \}$.

If we drop the above mentioned index-ordering restriction (thereby allowing, for instance, such an equation as $ x = \{x\} $), in order to guarantee the existence of solutions in the model  we need to work with   universes \textit{richer} than $ \HF $. The most natural (and minimal) among such universes is a close relative of $ \HF^0 $, goes under the name of  $ \HF^{1/2} $, the universe of (rational) \emph{hyper}sets (see \cite{Acz88,OmodeoPT2017}), and can be defined as the extension of $ \HF $ obtained postulating unique solution to \emph{all} finite systems of equations of the above form---no constraint on the indexes of the $ x_{i,j}$'s, that now can be one among $ x_1, \ldots, x_{n}$.  
An example of hypersets in $\HF^{1/2}$ is the set $y= \{\emptyset,\{\emptyset,\{\emptyset,\{ \ldots\} \} \}$ above, obtained as the solution of the system of equations $\{y= \{\emptyset,y\}\}$.

The universe $ \HF^{1/2} $ of rational hypersets  is the one we will mostly use. The elements in $ \HF^{1/2} $ are called \emph{rational} in analogy to rational numbers, and $ \HF^{1/2} $ can be further extended (to $ \HF^{1} $), admitting even hypersets characterised by infinite systems of set-theoretic equations (with a unique solution) only. 

A complete discussion relative to universes of sets that can be used as models of $ \Omega $ goes beyond the scope of this paper. However, it is convenient to point out that, in all cases of interest for us here, an especially simple view of  $ \Omega $-models can be given using \emph{finite graphs}. Actually, $ \HF^0 $ or $ \HF^{1/2} $ can be identified as the collection of finite graphs---either acyclic or cyclic, respectively---, where sets are nodes and arcs depict the membership relation among sets (see \cite{OmodeoPT2017}). Given one such  \textit{membership} graph $ G $, its nodes represent a (hyper)set $ s $ together with the elements of the \emph{transitive closure} of $ s $ (i.e. the elements of $ s $, the elements of the elements of $ s $, the elements of the elements of the elements $ s $, ... . See Definition \ref{def-transitive}). 

More precisely, an hereditarily finite set can be uniquely represented by a finite, acyclic, oriented and \emph{extensional} (different nodes  have different collections of  successors) graph, with the edge relation $ h \rightarrow h' $ standing for $ h \ni h' $.  Any such graph as a single source, called the \emph{point}  of the graph, and a single sink, that is the empty set $ \emptyset $. This gives us the alternative---more graph-theoretic---view of the model $ \HF^{0} $, whose domain is now the collection of finite, well-founded, oriented and extensional graphs.   

A formal definition of $ \HF^{1/2} $ can be given using the above outlined graph-theoretic rendering of $ \HF^{0} $:  simply drop  acyclicity  and replace  extensionality with the requirement that no two nodes of the graph are \emph{bisimilar} (see \cite{OmodeoPT2017} for the definition of bisimilarity relation and recall that $ \HF^{0} \subset \HF^{1/2}$).  We keep also the requirement that a \emph{point}---i.e. a node from which every other node is reachable---of the graph is provided.

\begin{definition}
The hereditarily finite hypersets $ \HF^{1/2} $ denote the $ \Omega $-model   $\emme=({\cal U}, \cdot^\emme)$ such that

- \ $ \cal U $ is the collection of pointed, finite, oriented graphs whose only bisimulation relation is the identity;

-  \  $ \cdot^{\emme} $ is defined for $ \in $ as follows: $ h'\in^{\mathcal M} h $ holds  when  $ h' $ is  the sub-graph whose point is  one of the successor of the point of $ h $. The remaining operators are interpreted following their definition.
\end{definition}

A final further enrichment of both $ \HF^0 $ and $ \HF^{1/2} $ is obtained by adding \emph{atoms} (sometimes called \emph{urelements}) to their domain universes. Atoms can be thought as pairwise distinct copies of the empty set, are going to be denoted by $ \mathbf{a}_1, \mathbf{a}_2, \ldots  $, and collectively represented by $ \mathbb{A}= \{\mathbf{a}_1, \mathbf{a}_2, \ldots  \} $. The resulting universes will be denoted by $ \HF^0(\mathbb{A}) $ and $ \HF^{1/2}(\mathbb{A}) $.
When considering a model $\emme$ of $\Omega$  over the atoms in $\mathbb{A}$, we mean that $\mathbb{A} \subseteq {\cal U} $. 

While $HF^0$, $ \HF^{1/2} $ and $ \HF^{1} $ only contain finite sets, other models of $\Omega$ may also admit infinite sets, such as the infinite set of natural numbers.

In the next section, we will regard the domain $\Delta$ of a DL interpretation as a (finite or infinite) transitive set in a universe of an $\Omega$-model, 
i.e. $\Delta$ will be a set of sets in (a universe of a model of) the theory $\Omega$ rather than as a set of individuals, as customary in description logics.

\begin{definition}\label{def-transitive}
A element $ x $ in an $ \Omega $-model is said to be a \emph{transitive set} if it satisfies the formula: $(\forall y\in x)(y \subseteq x)$.
\end{definition}
For example, the set $x=\{ a,b,c, \{a,b\}\}$  over $\mathbb{A}$ (with $\mathbb{A} = \{a, b, c,\ldots\}$) is transitive, while the set $x'=\{ a,c, \{a,b\}\}$ is not transitive, as $\{a,b\} \in x'$ but $\{a,b\} \not \subseteq x'$.
Both $x$ and $x'$ are well-founded sets, instead the set of equations:  
$\{y=\{a,b,x\}$,  $x=\{b,y\} \}$ defines a collection of hypersets represented by the graph in Figure \ref{hypersets-example},
including the hyperset $y= \{ a,b, \{b,y\}\}$, which is neither well-founded nor transitive ($y$ is an hyperset in $\HF^{1/2}(\mathbb{A}) $).
%

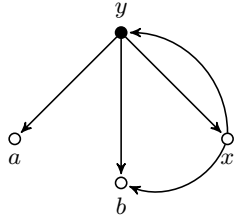
\begin{figure}
\begin{center}

\begin{tikzpicture}[->,>=stealth', shorten >=1pt, node distance=4cm, semithick, auto]
  \tikzstyle{every state}=[scale=.5]

  \node[state, fill, minimum size=1pt, label=above:{$ y $}] (Y)		 		{};
  \node[state, minimum size=1pt, label=below:{$ a $}] (A) [below left of=Y] 		{};
  \node[state, minimum size=1pt, label=below:{$ b $}] (B)  [below of =Y] 		{};
  \node[state, minimum size=1pt, label=below:{$ x $}] (X)  [below right of =Y] 		{};
  
 \path (Y) edge (A)
 (Y) edge (B)
 (Y) edge (X)
 (X) [bend left=45]  edge (B)
 (X) [bend right=45] edge (Y);
 
 \end{tikzpicture}
   
\caption{Example of a hyperset $ y $, built from atoms $ a $ and $ b $, which is neither well-founded nor transitive: $ x $ is an element of $ y $ but is not included in $ y $, since $ y $ (that belongs to the transitive closure of $ y $)  is not among the elements of $ y $.}\label{hypersets-example}
\end{center}
\end{figure}

\section{The description logic $\alc^\Omega$ } \label{sec:DL+Omega}

We start from the observation that 
in $\alc$  concepts are interpreted as sets (namely, sets of domain elements) 
and we generalize $\alc$ by allowing concepts to be interpreted as sets living in a model of the set theory $\Omega$.
In addition, we extend the language of $\alc$  by introducing  the power-set 
as a new concept constructor, and 
 allowing membership relations among concepts 
 in the knowledge base.
We call $\alc^\Omega$ the resulting extension of $\alc$.

As before, let $N_I$ , $N_C$, and $N_R$ be
the set of individual names, concept names, and role names in the language, respectively.
In building complex concepts, in  addition to the constructs of $\alc$, 
we also consider the difference $\backslash$ and the power-set $\texttt{Pow}$ constructs.
\begin{definition}The set of $\alc^{\Omega}$ \emph{concepts} are
defined inductively as follows:

- $A \in N_C$, $\top$ and $\bot$ are $\alc^{\Omega}$ \emph{concepts};
    
- if $C, D$ are $\alc^{\Omega}$ concepts and $R \in N_R$, then the following are $\alc^{\Omega}$ \emph{concepts}:
  \begin{center}
  $C \sqcap D, C \sqcup D, \neg C, C\backslash D, \texttt{Pow}(C),  \forall R.C, \exists R.C$ 
  \end{center}
\end{definition}
While the concept $C\backslash D$ can be easily defined as $C \sqcap \neg D$ in $\alc$,
this is not the case for the concept $\texttt{Pow}(C)$.
Informally, the instances of concept $\texttt{Pow}(C)$ are all the subsets of the instances of concept $C$, 
 which are ``visible'' in (i.e. which belong to) $\Delta$.
 
 Besides usual assertions of the forms $C(a)$ and $R(a,b)$ with $ a,b \in N_I $, $\alc^{\Omega}$ allows in the ABox
{\em concept membership axioms} and  {\em role membership axioms} of the forms $C \in D$ and $(C,D) \in R$, respectively, 
where $C$ and $D$ are $\alc^{\Omega}$ concepts and $R$ is a role name.

Considering again the example in Section 1, 
the additional expressivity of the language, in which general concepts (and not only concept names) can be instances of other concepts, allows for instance to represent the statement that  bears which are polar 
are in the  red list of endangered species, 
by axiom $\mathit{PolarCreature \sqcap Bear \in RedListSpecies}$.
We can further represent the fact the polar bears are more endangered than eagles by adding a role $\mathit{moreEndangered}$ and the role membership axiom
$\mathit{(PolarCreature \sqcap }$ $\mathit{ Bear, Eagle) \in moreEndangered}$.
The inclusion $\mathit{RedListSpecies \sqsubseteq}$ $\mathit{  \texttt{Pow}(CannotHunt)}$ means that any element of $\mathit{RedListSpecies}$ (such as $\mathit{Eagle}$) is a subset of $\mathit{CannotHunt}$, i.e., each single eagle cannot be hunted.
As shown in \cite{Motik05},  the meaning of the sentence 
``all the instances of species in the Red List are not allowed to be hunted"
could be captured by combining the $\nu$-semantics with 
SWRL \cite{Horrocks04}, but not by  the $\nu$-semantics alone. \normalcolor

We define a semantics for $\alc^\Omega$ by  extending the $\alc$ semantics in Section \ref{sec:ALC} 
to capture the meaning of concepts  (including concept $\texttt{Pow}(C)$) as elements (sets) of the domain ${\Delta}$, 
chosen to be a \textit{transitive} set (i.e. a set $ x $ such that $ x $'s elements are also $ x $'s subsets, see Definition \ref{def-transitive}).
Roles are interpreted as binary relations over the domain $\Delta$, concepts as subsets of $ \Delta $, and
individual names as elements of a set of atoms $\mathbb{A}$---from which the sets in $\Delta$ are built.

\begin{definition}
An interpretation for $\alc^\Omega$ is a pair $I=\langle \Delta, \cdot^I \rangle$ over a set of atoms $\mathbb{A}$ where:
\vspace{-0.2cm}
\begin{itemize}

\item the non-empty domain $\Delta$ is a transitive set chosen in the universe ${\cal U}$ of a model $\emme$ of $\Omega$  over the atoms in $\mathbb{A}$;

\item 
the extension function $\cdot^I$ maps each concept name $A\in N_C$ to a subset $A^I \subseteq \Delta$,\footnote{Observe that condition $A^I \subseteq \Delta$ is the usual one  for concept names in $\alc$ semantics, and it is weaker than the semantic condition $A^I \in \Delta$ required in  \cite{ICTCS_2018}  and in \cite{Jelia_2019}. This allows a more uniform treatment of all concepts and slightly simplifies the set-theoretic translations.}
each role name $R \in N_R$ to  a binary relation $R^I \subseteq  \Delta \times  \Delta$,
and each individual name $a \in N_I$ to an element $a^I \in \mathbb{A} \cap \Delta$.
\end{itemize}
The function $\cdot^I$ is extended to complex concepts of $\alc^\Omega$, as in Section \ref{sec:ALC} for $\alc$,
but for the two  additional cases:
$(\texttt{Pow}(C))^I = \mathit{Pow}(C^I) \cap \Delta$ \ and \  $(C \backslash D)^I= (C^I \backslash D^I) $. 
%
%
%
%
%
%
%
%
%
\end{definition}
Observe that $ \mathbb{A} \cap \Delta$ consists of the 
atoms in $\mathbb{A} $ necessary for the interpretation of individual names.
Moreover, even though ${\cal U}$ is closed under  union, power-set, etc., the set $\Delta$ is not guaranteed to be so.
In particular, the interpretation $C^I$ of a concept $C$ is not necessarily  an element of  $\Delta$. 
However, given the interpretation above of the power-set concept 
as the portion of the (set-theoretic) power-set {\em visible in $\Delta$}, it easy to see by induction that, for each $C$, the extension of $C^I$ is a subset of $\Delta$ (i.e., $C^I \subseteq \Delta$).

The requirement that the set $\Delta$ is transitive is needed to guarantee that, when we consider any set $x$ in $\Delta$, all the instances of $x$ are elements of $\Delta$ as well. For example, 
 if $Eagle^I$ is  an element of $\Delta$, then any specific element (any eagle) in $Eagle^I$ must be an element of $\Delta$. 
 
 As we will see later, the choice of $\Delta$ being equal to the universe ${\cal U}$ (rather than being a transitive set in  ${\cal U}$) is not viable when one wants to include in the language a concept  $\top$,  
as usual in description logics. Indeed, the universe ${\cal U}$ is not, in general, a set, while the interpretation of $\top$ must be a set (and all set theoretic operations, including the power-set, can be applied to it).

While ${\cal U}$  is always infinite, $\Delta$ is not necessarily an infinite set.  
Also, the interpretation of $\texttt{Pow}(\top)$ is not always the same as  that of $\top$.
Consider the following example with $\Delta$ a finite transitive set:
\[
\Delta =\{a,b,c,\{a,b\},\{a,c\}\}.
\]
By definition, $\top^I= \Delta$ and 
$\texttt{Pow}(\top)^I= Pow(\Delta) \cap \Delta =\{\{a,b\},\{a,c\}\}$.
Hence, in this example, $Pow(\top)^I \neq \top^I$
and, furthermore, $\top^I \not \in \Delta$ .

As a further observation, since extensionality does not hold in $\alc^\Omega$ (as it does not hold in $\Omega$) two concepts  $Eagle$ and $Aquila$ with the same extension may be interpreted as different sets in the models of the knowledge base (i.e., $Eagle^I \neq Aquila^I$), although they have the same elements. As a consequence, for instance, from  $Eagle^I \in RedListSpecies^I$, one cannot conclude $Aquila^I \in RedListSpecies^I$.

Given an interpretation $I$, the satisfiability of inclusions and assertions is defined as in  $\alc$ interpretations
(Definition \ref{satisfiability}). 
Satisfiability of (concept and role) membership axioms in an interpretation $I$ is defined as follows: {\em $I$ satisfies $C \in D$} if  $C^I \in D^I$; 
 {\em $I$ satisfies $(C,D) \in R$} if  $(C^I, D^I) \in R^I$. 
With this addition, the notions of satisfiability of a KB and of
entailment in $\alc^\Omega$ (denoted  $\models_{\alc^\Omega}$) can be defined as in Section \ref{sec:ALC}.

The problem of instance checking in $\alc^\Omega$ includes both the problem of verifying whether an assertion $C(a)$ is a logical consequence of the KB
and the problem of verifying whether a membership $C \in D$ is a logical consequence of the KB (i.e., whether $C$ is an instance of $D$).
\begin{example} \label{exa:readingGroup}
Let $K= ({\cal T}, {\cal A})$ be a knowledge base, where ${\cal T}$ is the set of inclusions:
\begin{quote}
(1) $\mathit{ReadingGroup \sqsubseteq \texttt{Pow}(Person)}$ \\
(2) $\mathit{Meeting \sqsubseteq \texttt{Pow}(ReadingGroup)}$ \\
(3) $\mathit{Meeting \sqsubseteq \texttt{Pow}(\exists has\_leader. Person)}$ \\
(4) $\mathit{SummerMeeting \sqsubseteq \texttt{Pow}( \exists has\_paid. Fee)}$ 
\end{quote}
and ${\cal A}$ contains the assertions (for conciseness, we write $A,B \in C$ for $A \in C$ and $B \in C$):
\begin{quote}
$\mathit{HistoryGroup, FantasyGroup, ScienceGroup \in ReadingGroup;  }$\\ 
$\mathit{SummerMeeting, WinterMeeting \in Meeting; }$\\ 
$\mathit{ScienceGroup, FantasyGroup \in SummerMeeting; }$\\ 
$\mathit{bob \in FantasyGroup; \; alice, bob \in ScienceGroup; \; carl \in HistoryGroup. \;}$
\end{quote}
Each reading group is a set of persons (1) and, in particular, the history, fantasy and science groups are reading groups.
Bob is in the Fantasy group Carl is in the History group, while Alice and Bob are in the Science group.
Each meeting is a set of reading groups (2). In particular, the $\mathit{SummerMeeting}$ and the $\mathit{WinterMeeting}$ are meetings.
Both the Science group and the Fantasy group participate to the $\mathit{SummerMeeting}$.
Each reading group in a meeting has a leader, who is a person (3).
All  participants to the $\mathit{SummerMeeting}$ have paid the registration fee (4).

From this specification, we can conclude that both the science and the fantasy groups have some leader who is a person, and that Alice and Bob
participate to the summer meeting and  have paid the registration fee.
For instance,
as $\mathit{SummerMeeting \in }$ $\mathit{ Meeting}$, by inclusion (3),  $\mathit{SummerMeeting \in \texttt{Pow}(\exists has\_leader. Person)}$, i.e.,
$\mathit{Sum}$- $\mathit{merMeeting  \sqsubseteq \exists has\_leader. Person}$.
Now, as $ScienceGroup \in SummerMeeting$, then $\mathit{ScienceGroup  \in \exists has\_leader. Person}$, and we can conclude that the Science group has a leader who is a person.

To see that Bob has paid the registration fee, consider that 
$\mathit{ bob \in ScienceGroup}$. As $\mathit{Science}$- $\mathit{Group \in SummerMeeting }$, by (4), 
$\mathit{ScienceGroup \in  \texttt{Pow}( \exists has\_paid. Fee) }$.
Then, $\mathit{ScienceGroup }$ $\mathit{ \sqsubseteq \exists has\_paid. Fee}$ and,
therefore,  $\mathit{bob \in \exists has\_paid. Fee}$.
Notice that, instead of axiom (4), we could have introduced the inclusion axiom 
$\mathit{Meeting \sqsubseteq}$ $\mathit{(\texttt{Pow}(\texttt{Pow}(\exists has\_paid. Fee))}$
meaning that, for any meeting, all participants have paid the registration fee.
\end{example}
\normalcolor
In the next section, we define a polynomial encoding of the language $\alc^\Omega$ 
into $\alcio$.

\section{Translation of $\alc^\Omega$ into $\alcio$} \label{sec:ALC^Omega to ALCIO}

To provide a proof method for $\alc^\Omega$, we define a translation of  $\alc^\Omega$ into the description logic $\alcio$,
including inverse roles and nominals. 
In \cite{DAgostino1995}  the membership relation $\in$ is used to represent the accessibility relation $R$ of a normal modal logic.    
In this section, vice-versa, we exploit the correspondence between $\in$ and the accessibility relation of a modality, by introducing
a new (reserved) role $e$ in $N_R$ to represent the inverse of the membership relation:
in any interpretation $I$,  $(x,y) \in e^I$ will stand for $y \in x$. 
The idea underlying the translation is that each element $u$ of the domain $\Delta$ in an $\alcio$ interpretation $I=\langle \Delta, \cdot^I \rangle$ 
can be regarded as the set of all the elements $v$ such that $(u,v) \in e^I$.

\medskip

The translation of a knowledge base $K= ({\cal T} ,{\cal A})$ of $\alc^\Omega$ into $\alcio$ can be defined as  follows.
First,  we associate each concept $C$ of $\alc^\Omega$ to a concept $C^T$ of $\alcio$ by replacing all occurrences of the power-set constructor
$\texttt{Pow}$ with a concept involving the universal restriction $\forall e$ (see below). 
More formally, we inductively define the translation $ C^T $  of $ C $ by simply recursively replacing every  subconcept
$\texttt{Pow}(D)$ appearing in $C$ by $\forall e. D^T$, while the translation $T $ commutes with concept constructors in all other cases  
 (and $ B^T= B $, for any concept name $B$). 

Semantically this will result in interpreting any (sub)concept $(\texttt{Pow}(D))^I$ by 
\[(\forall e.D)^I=\{x \in \Delta  \tc \forall y( (x,y) \in e^I  \imp y \in D^I)\},\]
which, recalling that $ (x,y)\in e^I $ stands for $ y \in x $,  characterises the collection of  subsets of $ D^I $ \emph{visible in $\Delta$} (i.e. subsets of  $D^I$ that are also elements of $\Delta$): $(\forall e.D)^I=\{x \in \Delta \tc \forall y( y \in x  \imp y \in D^I)\}$,
that is,  $(\forall e.D)^I=\{x \in \Delta \tc x \subseteq D^I)\}=\texttt{Pow}(D^I)\cap \Delta =(\texttt{Pow}(D))^I$, 
as expected.

\vspace{-0.2cm}
\subsection{Translating TBox, ABox, and queries} \label{sec: translation}
\vspace{-0.1cm}

We define a new TBox, ${\cal T}^T$, by introducing,
for each inclusion $C \sqsubseteq D$ in ${\cal T}$,
the inclusion $C^T \sqsubseteq D^T$ in ${\cal T}^T$.
Additionally,
for each (complex) concept $C$ occurring in the knowledge base $K$ (or in the query) on the l.h.s. of a membership axiom 
$C \in D$ or $(C,D) \in R$,
we extend $N_I$ with  
{\em a new individual name\footnote{The symbol $ e_C $ should remind the \emph{e}-xtension of $ C $.} $e_C $} and we add in ${\cal T}^T$  the concept equivalence: 
\begin{equation} \label{axiom:e_c}
C^T \equiv \exists e^-. \{e_C \}.  
\end{equation} 
From now on, new individual names such as $ e_C $ will be called \textit{concept individual names}. 
This equivalence is intended to capture the property that, in all the models $I=\langle \Delta, \cdot^I \rangle$ of $K^T$,
$e_C ^I$ is in relation $e^I$ with all and only the instances of concept $C^T$, i.e., for all $y \in \Delta$, $(e_C ^I,y) \in e^I$ if and only if $y \in (C^T)^I$.

As in the case of the power-set constructor, this fact can be verified by analyzing the semantics of $ \exists e^-. \{e_C \} $:
\[
(\exists e^-. \{e_C \})^I=\{x  \in \Delta \tc \exists y( (x,y) \in (e^-)^I  \wedge y \in (\{e_C \})^I\},
\]
which, recalling that $ e $ stands for $ \ni $ and interpreting the nominal, will stand for 
\[
(\exists e^-. \{e_C \})^I=\{x \in \Delta \tc \exists y( x \in y  \wedge y \in \{e_C ^I\}\}=\{x \in \Delta \tc  x \in e_C ^I \},
\]
which,
  by  concept equivalence (\ref{axiom:e_c}),
 is as to say that $ e_C ^I $ and $ (C^T)^I $  have the same extension. 

\begin{remark}
	It is important to notice that every  concept individual name  of the sort $ e_C  $ introduced above---that is, every individual name whose  purpose is that of providing a name to the extension
of  $ C^I $---, in general turns out to be in relation $ e $ with other elements of the domain $ \Delta $ of $ I $ 
(unless $C$ is an inconsistent  
concept and its extension is empty). 
This is in contrast with the  assumption relative to other ``standard'' individual names $ a \in N_I $, for which we will require $(\neg \exists e. \top)(a)$ (see below) as they are interpreted as atoms. 

\end{remark}
We define  ${\cal A}^T$ as the set of assertions containing: 

	- for each concept membership axiom $C \in D$ in ${\cal A}$,  the assertion $D^T(e_C )$,
	
	- for each role membership axiom $(C,D) \in R$ in ${\cal A}$,  the assertion $R(e_C,e_D)$, 
	
	- for each assertion $D(a)$ in ${\cal A}$,  the assertion $D^T(a)$,
	
	- for each assertion $R(a,b)$ in ${\cal A}$,  the assertion $R(a,b)$ and, finally,
	
	- for each (standard) individual name $a \in N_I$, 
	the assertion $(\neg \exists e. \top)(a)$.

\noindent
As noticed above, the last requirement  forces all named individuals (in the language of 
$K$) to be interpreted as domain elements which
are not in relation $e$ with any other element.

Let $K^T = ({\cal T}^T ,{\cal A}^T)$ be the knowledge base obtained by translating $K$ into $\alcio$. 

\begin{example}\label{exa:Eagle}

Let $K= ({\cal T} ,{\cal A})$ be the knowledge base considered above:
\\
${\cal T}= \{ \mathit{RedListSpecies \sqsubseteq \texttt{Pow}(CannotHunt)}\}$ and 
\\
${\cal A}=$ $\{\mathit{Eagle(harry), Eagle \in RedListSpecies, \; PolarCreature \sqcap Bear \in RedListSpecies}\} $.

\noindent
By the translation above, we obtain:
\\
${\cal T}^T= \{ \mathit{RedListSpecies \sqsubseteq \forall e. CannotHunt, }$

$\;$ \mbox{\ \ \ \ } $  \mathit{Eagle \equiv \exists e^-. \{e_{Eagle}\}  , \; PolarCreature \sqcap Bear  \equiv \exists e^-. \{e_{PolarCreature \sqcap Bear}\} 
\;}\}$  
\medskip

\noindent
${\cal A}^T=$ $\{\mathit{Eagle(harry),   RedListSpecies(e_{Eagle}), RedListSpecies(e_{PolarCreature \sqcap Bear}) ,}$ 

$\;$ \mbox{\ \ \ } $\mathit{ \; (\neg \exists e. \top)(harry)} \; \} $

\vspace{0.4mm}
\noindent
$K^T$ entails  $ \mathit{ CannotHunt}$$\mathit{(harry)}$ in  $\alcio$.
In fact, from $\mathit{RedListSpecies(e_{Eagle})}$ and $\mathit{Red}$- $\mathit{ListSpecies \sqsubseteq \forall e. CannotHunt}$, 
it follows that, in all  models of $K^T$, $e^I_{Eagle}$ $ \in\mathit{(\forall e. Cannot}$- $\mathit{Hunt)^I}$.
Furthermore, from $\mathit{Eagle \equiv \exists e^-. \{e_{Eagle}\}}$ \  and the assertion \  
$\mathit{Eagle(harry)}$, it follows that
$\mathit{(e^I_{Eagle}, harry^I) \in e^I}$ holds. Hence, $ \mathit{ harry^I} \in $  $\mathit{CannotHunt^I}$.
As this holds in all models of $K^T$, $ \mathit{ CannotHunt(harry)}$ is a logical consequence of $K^T$.
It is easy to see that $ \mathit{Eagle \sqsubseteq CannotHunt}$ follows from $K^T$ as well.
\end{example}

Let $F$ be a {\em query} of the form $C \sqsubseteq D$, $C(a)$ or $C \in D$ 
We assume that all the individual names, concept names and role names occurring in $F$ also occur in $K$ and we
define a translation $F^T$ of the query $F$ as follows:
\begin{quote}
	- if $F$ is a subsumption  $C \sqsubseteq D$, then $F^T$ is  the subsumption $C^T \sqsubseteq D^T$;
	
	- if $F$ is an assertion $C(a)$, then $F^T$ is the assertion $C^T(a)$;
	
	- if $F$ is a membership axiom $\mathit{C \in D}$ ($\mathit{(C,D) \in R}$, respectively), then $F^T$ is the assertion $D^T(e_C )$ ($\mathit{R(e_C,e_D )}$, respectively). 
	
\end{quote}
In the following we prove the  soundness and completeness of the translation of an ${\alc^\Omega}$ knowledge base into $\alcio$.


\begin{proposition}[Soundness of the translation]\label{Prop:Soundness-translation}
The translation of an $\alc^\Omega$  knowledge base $K= ({\cal T} ,{\cal A})$  into $\alcio$ is sound,  that is, 
for any query $F$:
\begin{align*}
	K^T \models_{\alcio} F^T & \Rightarrow  K \models_{\alc^\Omega} F.
\end{align*}

\end{proposition}
\begin{proof} (Sketch)
By contraposition, assume  $K \not \models_{\alc^\Omega} F$ and let 
$I=\langle \Delta, \cdot^I \rangle$ be a model of $K$ in ${\alc^\Omega} $ that falsifies $F$.
$\Delta$ is a transitive set living in a model of $\Omega$ with universe ${\cal U}$.
We build an $\alcio$ interpretation $I'=\langle \Delta', \cdot^{I'} \rangle$, which is going to be a model of $K^T$ falsifying $F$ in $\alcio$, by letting:

-  $\Delta'=\Delta$;  

- for all $B \in N_C$, $B^{I'}=B^{I}$;    

-  for all roles $R \in N_R$, $R^{I'}= R^I$;

-  for all $x,y \in \Delta'$, $(x,y) \in e^{I'}$ if and only if $y \in x$;

-  for all (standard) individual name $a \in N_I$, $a^{I'}=a^I \in \mathbb{A} \cap \Delta$;

-  for all $e_C  \in N_I$, $e_C ^{I'}=C^I.  $ \\
The interpretation $I'$ is well defined. First, the interpretation $B^{I'}$ of a named concept $B$ is a subset of $\Delta'$ as expected. 
In fact,  for each $x\in B^{I'}$,  $x\in B^I \subseteq \Delta = \Delta'$.
Also, $a^{I'}=a^I \in \Delta= \Delta'$.
It is easy to see that the interpretation of constant $e_C$, $e_C^{I'}$ is in $\Delta'$. In fact, 
as the named individual $e_C$ has been added by the translation to the language of $K^T$, there must be some membership axiom $C \in D$ (or $(C,D) \in R$) in $K$, for some $D$ (respectively, for some $D$ and $ R$).
Considering the case that axiom $C \in D$ is in $K$, as $I$ is a model of $K$, $I$ satisfies $C \in D$, so that $C^I \in D^I$ must hold.
However, as $D^I \subseteq \Delta$, it must be $C^I \in \Delta$.
Hence, by construction, $e_C ^{I'} = C^I \in \Delta'=\Delta$.
In case $(C,D) \in R$, it must hold that $(C^I,D^I) \in R^I$. 
As $R^I \subseteq \Delta \times \Delta$, then $C^I, D^I \in \Delta$.
In particular, $e_C ^{I'} = C^I \in \Delta'=\Delta$.

We can prove by induction on the structural complexity of the concepts that, 
for all $x \in \Delta'$,
\begin{equation} 
 x \in (C^T)^{I'} \mbox{ if and only if } x \in C^I
 \end{equation}
and show that all the axioms and assertions in $K^T$ are satisfied in $I'$, and 
$F^T$ is falsified in $I'$.
\hfill $\Box$
\end{proof}

\medskip

\noindent 
Before proving completeness of the translation of $\alc^\Omega$ into $\alcio$,
we show that, if the translation $K^T$ of a knowledge base $K$ 
has a model in $\alcio$, then 
it has a finite model.

\begin{proposition}\label{finite} \label{prop:finite_model_property}
	Let $K$ be a knowledge base in $\alc^\Omega$ and let $K^T$ be its translation in $\alcio$.
	If $K^T$ has a model in $\alcio$, then it has a finite model.
\end{proposition}
\begin{proof}
	We prove this result by providing an alternative (but equivalent) translation $K^{T(\neg)}$ of $K$ in the description logic  $\alc (\neg)$,
	using a single negated role $\neg e$.
	
	$\alc(\neg)$ extends $\alc$ with role complement operator, where, for any role $R$, 
	the role $\neg R$ is  {\em  the negation of role $R$}, where $(x,y) \in (\neg R)^I$ if and only if $(x,y) \not \in R^I$.
	In the translation, we exploit $\neg e$ to capture non-membership, where
	$(x,y) \in (\neg e)^I$ if and only if $(x,y) \not \in  e^I$ (i.e., in  set terms, $y \not \in x$).
	Decidability of concept satisfiability in $\alc(\neg)$ has been proved by Lutz and Sattler in \cite{LutzSattler2001}.
	The finite model property of a language with a single negated role $\neg e$ can be proved as done in  \cite{Gargov_etal_1987} (Section 2)
	for a logic with the ``window modality",
	by standard filtration, extended to deal with additional K-modalities (for the other roles) as in the proof in \cite{Blackburn2001}.
	Indeed, as observed in \cite{LutzSattler2001}, the ``window operator" $\window$ studied in \cite{Gargov_etal_1987} is strongly related to a negated modality, as
	$\window \phi$ can be written as $[\neg R] \neg \phi$.

	The translation $K^{T(\neg)}$ can be defined modifying $K^T$ by replacing the concept equivalence $C^T \equiv \exists e^-. \{e_C \}$ with the  assertions:
	$ (\forall e. C^T)(e_C )$ and $ (\forall (\neg e).( \neg C^T))(e_C )$.
	
	One can show that from any model $I=(\Delta, \cdot^I)$ of $K^{T(\neg)}$  we can easily define a model of $K^T$ in $ \alcio$, and vice-versa---considering the usual interpretation of negated roles, inverse roles and nominals. 
	In fact,
	the semantics of the assertion $ (\forall e. C^T)$ $(e_C )$ is the following: 
	for all $x \in \Delta$, $(e_C ^I,x) \in e^I \Rightarrow x \in (C^T)^I$, 
	which is equivalent to  $\exists e^-.\{e_C \} \sqsubseteq C^T$.

	The semantics of the assertion $ (\forall (\neg e).( \neg C^T))(e_C )$ is: 
	for all $x \in \Delta$,  $(e_C ^I,x) \not \in e^I \Rightarrow x \not \in (C^T)^I$, i.e., 
	for all $x \in \Delta$, $ x \in (C^T)^I \Rightarrow (e_C ^I,x) \in e^I $,
	which is the semantics of  $C^T \sqsubseteq \exists e^-.\{e_C \}$.
	
	We conclude the proof by observing that, if $K^T$ has a model in $\alcio$, there is a model of $K^{T(\neg)}$.
	Then, by the finite model property of $\alc(\neg)$, $K^{T(\neg)}$  must have a finite model, from which  a finite model of $K^T$ can be defined. 
	\hfill $\Box$
\end{proof}
As a byproduct of the above proposition, we have that any $\alc^\Omega$ knowledge base $K$ has a translation $K^{T(\neg)}$ 
in the description logic  $\alc (\neg)$, which uses a single negated role $\neg e$,
and that each model of $K^{T(\neg)}$ can be mapped to a corresponding $\alcio$ model of  $K^{T}$, and vice-versa.

%
%
To conclude our analysis we now prove the completeness of the translation $K^{T}$ of a knowledge base $K$ in $\alcio$.

\begin{proposition}[Completeness of the translation]\label{Prop:Completeness-translation}
The translation of an $\alc^\Omega$  knowledge base $K= ({\cal T} ,{\cal A})$  into $\alcio$ is complete,  that is, 
for any query $F$:
\begin{align*}
	K \models_{\alc^\Omega} F & \Rightarrow K^T \models_{\alcio} F^T.
\end{align*}
\end{proposition}
\begin{proof}
We prove the completeness of the translation by contraposition. 
Let  $K^T \not \models_{\alcio} F^T$.  
Then there is a model 
$I=\langle \Delta, \cdot^I \rangle$ of $K^T$ in $\alcio$ such that $I$ falsifies $F$.
We show that we can build a model $J=\langle \Lambda, \cdot^J \rangle$ of $K$ in $\alc^\Omega$,
where
the domain $\Lambda$ is a transitive set
in the universe $\HF^{1/2}(\mathbb{A})$ consisting of all the hereditarily finite rational hypersets
built from atoms in $\mathbb{A}=\{{\bf a_0}, {\bf a_1},\ldots  \}$.
As a matter of fact, the domain $ \Lambda $ is to be extended possibly duplicating sets representing extensionally equal but pairwise distinct sets/elements in $\Delta$.

We define $ \Lambda $ starting from the graph\footnote{Strictly speaking the graph $ G $ introduced here is not really necessary: it is just mentioned to single out the membership relation $ \in $  from $ e^I $ more clearly.} $ G=\langle \Delta, e^I \rangle$, whose nodes are the elements of $ \Delta $ and whose arcs are the pairs $ (x,y)\in e^I $. Notice that, by Proposition \ref{finite}, the graph $ G $ can be assumed to be finite. 
Intuitively, an arc from $ x $
to $ y $ in $ G $ stands for the fact that $ y \in x $.

At this point, let $\Delta_0= \{ d_1, \ldots , d_m\} $ be the elements of $\Delta$ which, in the model $I=\langle \Delta, \cdot^I \rangle$, are not in relation $e^I$ with any other element in $\Delta$ and are non equal to the interpretation of any concept individual name $e_C$ 
(that is, $d_j \in \Delta_0 $ iff there is no $y$ such that $(d_j, y) \in e^I$ and there is no concept $C$ such that $d_j= e_C^I$).
We define the $ M(d) $'s for $ d \in \Delta $,  as the hypesets satisfying the following collection of set-theoretic equations:
\begin{align}\label{def_M(d)}
	M(d) & = \left\{\begin{array}{ll}
						{\bf a_k} & \mbox{ if } d = d_k \in \Delta_0,\\
						\left\{M(d') \tc (d,d')\in e^I\right\} & \mbox{ otherwise. } 
					\end{array}\right.
\end{align} 
Observe that,  for the concepts $C$ occurring as l.h.s. of membership axioms, as axiom $ C^T= \exists e^-.\{e_C\} $ is satisfied in the model $I$ of $K^T$, it holds that $d' \in (C^T)^I$ iff $(e_C^I, d') \in e^I$. Therefore, 
for $d=e_C^I$, $M(d)= M(e_C^I)= \left\{M(d') \tc (e_C^I,d')\in e^I\right\} $ $ = \left\{M(d') \tc d'\in (C^T)^I\right\}$.

The above definition uniquely determines hypersets in $ \HF^{1/2}(\mathbb{A}) $. This follows from the fact that all finite systems of (finite) set-theoretic equations have a solution in $ \HF^{1/2}(\mathbb{A}). $\footnote{When  $ e^I $ is a well-founded relation, $ M(\cdot) $ is its inductively defined set-theoretic ``rendering'' going under the name of \emph{Mostowski collapse} of $ e^I $ (see \cite{Jech:2003ly}).  As a consequence of the duplication of extensionally equal sets, not only we have the trivial property that, for $d,d' \in \Delta$, $d=d'$ implies $M(d)=M(d')$, but also the converse implication, i.e., 
$M(d)=M(d')$  implies $d=d'$.}

To complete the definition of $ J= \langle \Lambda, \cdot^J\rangle $ in such a way to prove that $ J $ is a model of $ K $ in $\alc^\Omega$ falsifying $ F $, we put:

	-  $\Lambda=\{M(d) \tc d \in \Delta\}$;    
	
	-  for all $B \in N_C$, $B^{J}=\{M(d) \tc  d \in B^I\};$
	
	-  for all roles $R \in N_R$ such that $ R \neq e $, 
		$R^{J}= \{(M(d),M(d')) \tc (d,d')\in R^I\};$
	
	-   for all standard named individuals $a \in N_I$ such that $a^I= d_k$, let $a^J= M(d_k) ={\bf a_k}\in \mathbb{A}$.\\
By construction,  $\Lambda$ is transitive set in a model $\emme$ of $\Omega$ 
(in fact, for all $M(d) \in \Lambda$, if $M(d') \in M(d)$, then $(d', d) \in e^I$ and then $d' \in \Delta$; therefore, $M(d') \in \Lambda$).
Moreover, it can be proved, by induction on the structural complexity of concepts, that, for all $x \in \Delta$: 
\begin{align}\label{set-eq_testo}
	M(x) \in C^J & \mbox{ if and only if } x \in (C^T)^I.
\end{align}
Let us consider the interesting case: $ C=\texttt{Pow}(D)  $. By definition of $ ~T $ , we have that: 
\begin{align*}
(C^T)^I & = ((\texttt{Pow}(D))^T)^I = (\forall e.D^T)^I  =\{x \in \Delta \tc \forall y( (x,y) \in e^I  \imp  y \in (D^T)^I\}
\end{align*}
and 
$C^J  =(\texttt{Pow}(D))^J=\mathit{Pow}(D^J) \cap \Lambda$.
Consider, for $ x \in \Delta $, $ M(x)\in \mathit{Pow}(D^J) \cap \Lambda $, which is as to say that $ M(x)\subseteq D^J$. 
All the elements of $ M(x) $ are of the form $ M(y) $ for some $ y\in \Delta $, therefore we have that: 
\begin{align*}
\forall M(y)(M(y)\in M(x) \rightarrow M(y) \in D^J), 
\end{align*}
which, by definition of $ M(\cdot) $ and by inductive hypothesis, means that: 
\begin{align*}
\forall y((x,y) \in e^I \rightarrow y \in (D^T)^I), 
\end{align*}
which means $x \in (\forall e.D^T)^I =((\texttt{Pow}(D))^T)^I$
and proves (\ref{set-eq_testo}) in this case.

From  (\ref{set-eq_testo}) it is easy to prove that  all axioms and assertions in $ K $ are satisfied in $ J $.
\hfill $\Box$
\end{proof}
As the translation of $\alc^\Omega$ into $\alcio$  is polynomial (actually, linear) in the size of the knowledge base (and of the query) the following complexity result follows. 

\begin{proposition}
Concept satisfiability in $\alc^\Omega$  is an  \textsc{ExpTime}-complete problem.
\end{proposition}
The hardness comes from the \textsc{ExpTime}-hardness of concept satisfiability in $\alc$ with general TBox \cite{Schild91,handbook}. 
The upper bound comes from the \textsc{ExpTime} upper bound for ${\cal SHOI}$ \cite{HladikIJCAR2004}.

	\section{A set theoretic translation of  $\alc^\Omega$} \label{sec:tr_DL to sets}

To translate $\alc^\Omega$ in  $\Omega$, we exploit the polymodal version of 
the correspondence between $\in$ and the accessibility relation of a normal modality
studied in \cite{DAgostino1995} and used above. We start modifying our previously introduced translation of $\alc$ along the line used by D'Agostino et al. 
to deal with normal, complete finitely axiomatizable polymodal logics \cite{DAgostino1995}.
Then, we use the well known correspondence between description logics and modal logics studied by Schild \cite{Schild91}, 
where concepts (sets of elements) play the role of propositions (sets of worlds) in the polymodal  logic, while 
universal and existential restrictions $\forall R$ and $\exists R$
play the role of universal and existential modalities $\Box_i$ and $\Diamond_i$.

In Section \ref{sec:tr_ALC^Omega-R to sets} we focus on the fragment of $\alc^\Omega$ admitting no roles, no individual names and no existential and universal restrictions, that we call $\al^\Omega$. We show that $\al^\Omega$ can be given a simple set-theoretic translation 
in $\Omega$. 
Finally, in Section  \ref{sec:encoding}, we see that this set-theoretic translation 
can  be naturally extended to the full $\alc^\Omega$.
In particular, we encode $\alc^\Omega$ into its fragment $\al^\Omega$,
showing that $\al^\Omega$ is as expressive as $\alc^\Omega$ and providing a set-theoretic translation of $\alc^\Omega$
in which $\forall R_i.C$ and the power-set concept $\texttt{Pow}(C)$  are encoded in a uniform way.

	\subsection{A set theoretic translation of  $\alc$ with empty ABox } \label{sec:tr_ALC to sets}

Let $R_1,\ldots,R_k$ be the roles occurring in the knowledge base  $K= ({\cal T}, {\cal A})$ 
and let $A_1, \ldots, A_n$ be the concept names occurring in $K$.
Given a concept $C$ of $\alc$, built from the concept names and role names in $K$, 
its set-theoretic translation is a set-theoretic term $C^S(x,y_1,\ldots, y_k, x_1,\ldots,$ $ x_n)$, where $x,y_1,\ldots, y_k, x_1,\ldots, x_n$ are set-theoretic variables, inductively defined as follows:
\begin{quote}
   $\top^S=x$;  \ \ \ \ \ \ \ \ \ \ \ \ \ \ \ \ \ \ \ \ \ \ \ \ \ \ \ \ \ \ \ \ \ \ \ \ \ \ \ \ \ \ $\bot^S=\emptyset$;
   
   $A_i^S= x_i$ , for $A_i $  in $K$;   \ \ \ \ \ \ \ \ \ \ \ \ \ \ \ \ \ \ \  $(\neg C)^S=x \backslash C^S$;
  
   $(C \sqcap D)^S=C^S \cap D^S$;   \ \ \ \ \ \ \ \ \ \ \ \ \ \ \ \ \ \    $(C \sqcup D)^S=C^S \cup D^S$;
  
   $(\forall R_i.C)^S= Pow(((x \cup y_1 \cup \ldots \cup y_k) \backslash y_i ) \cup Pow(C^S))$, \ \ \ for $R_i$  in $K$;
 
\end{quote}
$(\exists R_i.C)^S$ is translated to the set-theoretic term $(\neg \forall R_i. \neg C)^S$.
Each $\alc$ concept $C$ is represented by a set-theoretic term $C^S$ and interpreted as a set in each model of $\Omega$.
Membership is used to give an interpretation of roles, as for modalities in the polymodal logics in \cite{DAgostino1995}.

For a single role $R$, by \textit{imitating} the relation $R^I$ with $\in$ (where  $v \in u$ corresponds to $(u,v) \in R^I$), we naturally obtain that $\texttt{Pow}(C)$ corresponds to the universal restriction $\forall R.C$.
For multiple roles, in order to encode the different relations $R_1,\ldots,R_k$, $k$ sets $U_i$ 
are considered.
Informally, each set $U_i$ (represented by the variable $y_i$) is such that $(v,v') \in R_i^I$ iff there is some $u_i \in U_i$ such that $u_i \in v$ and  $v' \in u_i$.

Given an $\alc$ knowledge base $K=({\cal T}, {\cal A})$ with $ {\cal A}= \emptyset$, we define the translation of the TBox axioms as follows:
\begin{quote}
$\mathit{TBox}_{\cal T}(x,y_1,\ldots,y_k, x_1, \ldots, x_k)= \{  C_1^S \cap x \subseteq C_2^S \mid C_1 \sqsubseteq C_2  \in {\cal T} \}$
 \end{quote}
We can then establish a correspondence between subsumption in $\alc^\Omega$ and derivability in the set theory $\Omega$,
instantiating the result of Theorem 5 in \cite{DAgostino1995} as follows:

\begin{proposition}  \label{set-th-translation-alc} 
For all concepts  $C$ and $D$ on the language of the theory $K$:

$K \models_{\alc} C \sqsubseteq D$  { if and only if }

$\Omega \vdash \forall x \forall y_1 \ldots \forall y_k (Trans^2(x)$
$\rightarrow \forall x_1, \ldots, \forall x_n (  \bigwedge \mathit{TBox}_{\cal T}  \rightarrow C^S \cap x \subseteq D^S))$\\
where $Trans^2(x)$ stands for $\forall y \forall z (y \in z \wedge z \in  x  \rightarrow y \subseteq x)$, that is, $x \subseteq \texttt{Pow}(\texttt{Pow}(x))$.\end{proposition}
The validity of $Trans^2(x)$ on the set $x$, which here represents the domain $\Delta$ of an $\alc$ interpretation,
 is required, as in the polymodal case  in \cite{DAgostino1995},  to guarantee that elements accessible through $R_i$ turn out to be themselves in $x$.

{ Roughly speaking, the meaning of Proposition  \ref{set-th-translation-alc} is that, for all  instances of $x$ representing 
the domain $\Delta$ and 
for all the instances $U_1, \ldots, U_k$ of the set variables $y_1,\ldots,y_k$,
any choice for the  interpretation $ x_1, \ldots, x_n$ of the concept names $A_1, \ldots, A_n$ in $K$ which satisfies the TBox axioms over the elements in $x$ (i.e., over the domain $\Delta$),
also satisfies the inclusion $C^S\subseteq D^S$ over $\Delta$.}

From the correspondence of the logic $\alc$ with the normal polymodal logic $K_m$ in \cite{Schild91}
and from the soundness and completeness of the set-theoretic translation for normal polymodal logics (Theorems 17 and 18 in \cite{DAgostino1995}),
we can conclude that, for $\alc$, the set-theoretic translation above is sound and complete.

This set-theoretic translation can then  be 
extended to other constructs of description logics 
including a set of axioms
$Axiom_H(x,y_1,\ldots,y_k)$, as  in \cite{DAgostino1995} to provide the translation of the specific axioms 
of a polymodal logic, as follows:

$\;$ \ \ \ \  $\forall x \forall y_1 \ldots \forall y_k (Trans^2(x) \wedge Axiom_H(x,y_1,\ldots,y_k) $

$\;$ \ \ \ \ \ \ \ \ \ \ \ \ \ \ \ \  \ \ \ \ \ \ \ \ \  \ \ \ \ \ \ \ \ \ \ \ \ \ \ \ \ \ \ \ \ \ \ \ \ 
$\rightarrow \forall x_1, \ldots, \forall x_n (  \bigwedge \mathit{TBox}_{\cal T}  \rightarrow C^S \cap x \subseteq D^S))$\\
In the following, we consider a few examples and we let for future work the development of a 
set-theoretic characterizations for more expressive DLs.


{\em  Role hierarchy} axioms have the form $R_j \sqsubseteq R_i$, and semantic condition $R_j^I \subseteq R_i^I$. 
They can be captured by adding in $Axiom_H(x,y_1,\ldots,y_k)$ the condition $y_j \subseteq y_i$.

{\em Transitivity of a role $R_i$},  $Trans(R_i)$, 
which corresponds  to the role inclusion axiom $R_i \circ R_i \sqsubseteq R_i$,  
can be captured 
adding the following axiom in $Axiom_H(x,y_1,\ldots,y_k)$:

$\forall y, u, v, u', z ( y \in x \rightarrow (( u\in y  \wedge u \in y_i \wedge v \in u  \wedge  u' \in v \wedge u' \in y_i \wedge z \in u')$

$\;$\ \ \ \ \ \ \ \ \ \ \ \ \ \ \ \ \ \ \ \ \ \ \ \ \ \ \ \ \ \ \ \ \ \ \ \ \ \ \ \ \ \ \ \ \ \ \ \ \ \ \ \ \ \ \ \ \ \ \  $\rightarrow \exists  u'' (  u''\in y  \wedge u'' \in y_i \wedge z \in u'')) )$

\noindent
encoding the semantic property 
$\forall y,v,z ( (y,v) \in R_i^I \wedge (v,z) \in R_i^I \rightarrow (y,z) \in R_i^I)$.

{\em Inverse roles}: 
Let a role $R_j$  be the inverse of $R_i$ (i.e., $R_j= R_i^-$).  
The semantic condition 
$(v,y) \in R_j^I$ if and only if $(y,v) \in R_i^I$ can be encoded by the axiom:

$\forall y, v (y \in x \rightarrow ( \exists u (u\in y  \wedge u \in y_j \wedge v \in u ) \leftrightarrow \exists u'( u'\in v  \wedge u' \in y_i \wedge y \in u' )))$\\
A similar axiom can be defined for complex role inclusions. A direct translation of nominals, $\{a\}$, would require a set theory with singleton operators.

	\subsection{A set-theoretic translation of the fragment $ \al^\Omega $} \label{sec:tr_ALC^Omega-R to sets}  

In this section we focus on the fragment $ \al^\Omega $ of $\alc^\Omega$ 
without roles, individual names, universal and existential restrictions and role assertions,
and we show that it can be given a simple set-theoretic translation 
in  $\Omega$. This translation provides some insight on the nature of the power-set construct in $\alc^\Omega$.

We call $\al^\Omega$ the fragment,  
whose concepts are defined inductively as follows: 

- $A \in N_C$, $\top$ and $\bot$ are $\al^{\Omega}$ \emph{concepts};
    
- if $C, D$ are $\al^{\Omega}$ concepts, then the following are $\al^{\Omega}$ \emph{concepts}:
  \begin{center}
  $C \sqcap D, C \sqcup D, \neg C, C\backslash D, \texttt{Pow}(C)$ 
  \end{center}

\noindent
The semantics of concept constructs in $\al^\Omega$ is the same as in $\alc^\Omega$.
An $\al^\Omega$ knowledge base $K$ is a pair $({\cal T}, {\cal A})$, where the TBox ${\cal T}$ is a set of concept inclusions $C \sqsubseteq D$, and the ABox ${\cal A}$ is a set of membership axioms $C \in D$. 
The notions of satisfiability of a knowledge base $K$ and entailment form $K$ are defined as in  $\alc^\Omega$.

Given an $\al^\Omega$ knowledge base  $K= ({\cal T}, {\cal A})$,
let $A_1, \ldots, A_n$ be the concept names occurring in $K$. 
We define a translation of an $\al^\Omega$ concept $C$ over the language of $K$ to a set-theoretic term  
$C^S(x, x_1,\ldots, x_n)$, where $x, x_1,\ldots, x_n$ are set-theoretic variables,
by induction on the structure of concepts, as follows:
\begin{center}
$\top^S=x;  ~ \bot^S=\emptyset; ~ A_i^S= x_i, \text{ for } i=1,\ldots, n;$\\
$	  (\neg C)^S=x \backslash C^S; ~   (C \sqcap D)^S=C^S \cap D^S;  ~ (C \sqcup D)^S=C^S \cup D^S; ~ (C \backslash D)^S=C^S \backslash D^S;$\\
$	    (\texttt{Pow}(C))^S = \mathit{Pow}( C^S). $
\end{center}
Let $K=({\cal T}, {\cal A})$.
The translation for the TBox ${\cal T}$ and  ABox ${\cal A}$ is defined as  follows:

$\mathit{TBox}_{\cal T}(x, x_1, \ldots, x_n)=  \{  C_1^S \cap x \subseteq C_2^S \mid C_1 \sqsubseteq C_2  \in {\cal T} \}$ 

$ \mathit{ABox}_{\cal A}(x,x_1, \ldots, x_n)=  \{C_1^S \in C_2^S \cap x \mid (C_1 \in C_2)  \in {\cal A} \} $\\
We can now establish a correspondence between subsumption in $\al^\Omega$ and derivability in 
$\Omega$.
\begin{proposition}[Soundness and Completeness of the translation of $\al^\Omega$]  \label{set-th-translation} 
For all concepts  $C$ and $D$ on the language of the knowledge base $K$: 

$K \models_{\al^\Omega} C \sqsubseteq D$  { if and only if } 

$\Omega \models  \forall x  (  Trans(x) $ 
$\rightarrow \forall x_1, \ldots, \forall x_n ( \bigwedge \mathit{ABox}_{\cal A} \wedge   \bigwedge \mathit{TBox}_{\cal T}  \rightarrow C^S \cap x \subseteq D^S))$

\noindent
where $Trans(x)$ stands for $\forall y (y \in x  \rightarrow y \subseteq x)$, that is, $x \subseteq \texttt{Pow}(x)$.

\end{proposition} 
%
\begin{proof}(Sketch)
($\Rightarrow$) By contraposition, suppose there is a model $\emme$ of $\Omega$, with universe $\U$ over $\mathbb{A}$, which falsifies the formula:
$\forall x  (  Trans(x) 
\rightarrow \forall x_1, \ldots, \forall x_n ( \bigwedge \mathit{ABox}_{\cal A} \wedge   \bigwedge \mathit{TBox}_{\cal T}  \rightarrow C^S \cap x \subseteq D^S))$.
Then there must be some $u \in \U$, such that 
$Trans(x)$ $[u/x]$ is satisfied in $\emme$, while 
$( \forall x_1, \ldots,$ $ \forall x_n ( \bigwedge \mathit{ABox}_{\cal A} \wedge   \bigwedge \mathit{TBox}_{\cal T}  \rightarrow C^S \cap x \subseteq D^S)))[u/x]$ is falsified in $\emme$.

Hence, there must be  $v_1, \ldots, v_n$ in $\U$, such that 
$( \bigwedge \mathit{ABox}_{\cal A} \wedge   \bigwedge \mathit{TBox}_{\cal T} )[u/x,  \overline{v}/ \overline{x}]$
is satisfied in $\emme$, while $ (C^S \cap x \subseteq D^S) [u/x, \overline{v}/ \overline{x}]$ is falsified in $\emme$.
Let $\beta=[u/x, \overline{v}/ \overline{x}]$.

We define an $\al^\Omega$ interpretation $I=(\Delta, \cdot^I)$, as follows:
%
 $\Delta=u$;  
%
$A_i^I=  v_i \cap u$, for all $i=1, \ldots, n$ such that $A_i$ occurs in $K$; and $A^I=  \emptyset$ for all other $A \in N_C$.  

\noindent
$I$ is well-defined.
By construction, $\Delta$ is a transitive set  living in the universe $\U$ of the $\Omega$ model $\emme$,
and $A_i^I \subseteq \Delta$.
We can prove by structural induction that, for all the concepts $C$ built from the concept names in $K$,
for the variable substitution $\beta= [u/x, \overline{v}/ \overline{x}]$, and
for all $w \in \Delta$:
\begin{align}\label{correspondence_C_C*_al}
w \in C^I  \mbox{ if and only if }  w \in (C^S)^\emme_\beta  
\end{align}
This equivalence can be used to prove that the $\alc^\Omega$ interpretation $I$ is a model of $K$, 
which falsifies the inclusion $C \sqsubseteq D$. Hence, $K \not \models_{\alc^\Omega} C \sqsubseteq D$.


($\Leftarrow$) By contraposition,
let  $I=(\Delta, \cdot^I)$ be $\al^\Omega$ model of $K$, falsifying the inclusion $C \sqsubseteq D$.
By construction, $\Delta$ is a transitive set living in the universe $\U$ of an $\Omega$ model $\emme$.
We show that $\emme$ falsifies the formula:
\begin{align} \label{formula_set_th_4.2_testo}
\forall x  (  Trans(x) \rightarrow \forall x_1, \ldots, \forall x_n ( \bigwedge \mathit{ABox}_{\cal A} \wedge   \bigwedge \mathit{TBox}_{\cal T}  \rightarrow C^S \cap x \subseteq D^S))
\end{align}
Let $\beta$ be the variable substitution $\beta= [u/x,  \overline{v}/ \overline{x}]$,
where: $u= \Delta$ and  $v_j= A_j^I $,  for all $j=1,\ldots,n$.
%
We can prove, by induction on the structure of the concept $C$, that for all the concepts $C$ built from the concept names in $K$,
and for all $d \in \Delta$:
\begin{align*}
d \in C^I  \mbox{ if and only if }  d \in (C^S)^\emme_\beta  
\end{align*}
This equivalence 
can be used to prove that  the formula (\ref{formula_set_th_4.2_testo}) is falsified in $\emme$, 
by showing that: 
$ ( \bigwedge \mathit{ABox}_{\cal A} \wedge   \bigwedge \mathit{TBox}_{\cal T})^\emme_\beta$
is satisfied in $\emme$ and that 
$(C^S \cap x \subseteq D^S)^\emme_\beta$ is falsified in $\emme$.
\hfill $\Box$
\end{proof}
A similar correspondence result can be proved for  {\em instance checking},
by replacing the inclusion $C^S  \cap x \subseteq D^S$ in Proposition \ref{Prop:encoding} with  $C^S  \in D^S \cap x$.

As we can see from the translation above, the power-set construct in $\al^\Omega$ 
is defined precisely as the set-theoretic power-set. 
From the translation it is clear that only the part of the power-set which is in the set $x$ (the domain $\Delta$) is relevant when evaluating the axioms in $K$ or  a  query.
In particular,  the axioms in the knowledge base
are only required to be satisfied over the elements of the transitive set $x$. 
Notice that it is the same in the set-theoretic translation of $\alc$ in Section \ref{sec:tr_ALC to sets}: knowledge base axioms are required to be satisfied on the elements of $x$.

Observe also that,  in both the translations of $\alc$ and of  $\al^\Omega$, $\top$ is interpreted as the transitive set $x$.
It would not be correct to interpret $\top$ as the universe $\U$ of a model of $\Omega$, as 
$\U$ is not a set. 
 In fact, $\texttt{Pow}(\top)$ is in the language of concepts 
and $(\texttt{Pow}(\top))^I= Pow(\top^I) \cap \Delta$. However, $Pow(\top^I)$ is not defined for $\top^I= {\cal U}$,
as ${\cal U}$ is not a set. \normalcolor

	\subsection{Translating $\alc^\Omega$ by encoding into $\al^\Omega$ } \label{sec:encoding}

In this section 
we  show that  $\al^\Omega$ has the same expressive power as $\alc^\Omega$,
as universal and existential restrictions of the language $\alc^\Omega$ (as well as role assertions)
can be encoded into $\al^\Omega$.
This encoding, together with the set-theoretic translation of $\al^\Omega$ given in the previous section, 
determines  a set-theoretic translation for $\alc^\Omega$,  in which roles are---ultimately---translated as in the polymodal translation in \cite{DAgostino1995}, and the power-set construct is translated accordingly.

Given an $\alc^\Omega$ knowledge base  $K= ({\cal T}, {\cal A})$,
let $R_1,\ldots,R_k$ be the role names occurring in $K$, $A_1, \ldots, A_n$ the concept names occurring in $K$,
and $a_1, \ldots, a_r$ the individual names occurring in $K$.
We introduce $k$ new concept names  $U_1, \ldots, U_k$ in the language, one for each role $R_i$.
These concepts (which are not concept names in $K$) will be used to encode universal restrictions $\forall R_i. C$ as well as the power-set concept $\texttt{Pow}(C)$ of $\alc^\Omega$ into $\al^\Omega$.
We further introduce  a new concept name $B_i$ for each individual name $a_i$ occurring in $K$,
a new concept name $F_{h,j}^i$ for each role assertion $R_i(a_h,a_j)$  
and a new concept name $G_{C_h,C_j}^i$ each role membership axiom $R_i(C_h,C_j)$. 

For an  $\alc^\Omega$ concept $C$,  
the encoding $C^E$ in $\al^\Omega$ can be defined by recursively replacing:
every named individual $a_i$ with the new concept name $B_i$, 
every subconcept $\forall R_i.C$  with $(\forall R_i.C)^E$ and 
every subconcept  $\texttt{Pow}(C)$ with  $(\texttt{Pow}(C))^E$, as defined below,
while the encoding $E$ commutes with concept constructors in all other cases. In particular, we let:

$\bullet$  $a_i^E= B_i$, for all $i=1,\ldots,r$;

$\bullet$ 
    $A^E= A$, for all concept names $A$ occurring in $K$;    
 
$\bullet$        
    $(\forall R_i.C)^E= \texttt{Pow}(\neg  U_i  \sqcup  \texttt{Pow}(C^E))$ 
    
$\bullet$     
    $(\texttt{Pow}(C))^E = \texttt{Pow}(  U_1 \sqcup \ldots \sqcup U_k   \sqcup  C^E)$ 

%
%
\noindent
The encoding  of $\forall R_i.C$ is based on the same idea as the set-theoretic encoding of $\forall R_i.C$ in Section \ref{sec:tr_ALC to sets}.
For each $(\forall R_i.C)^E$-element $y$ we require that, for all $y' \in y$ which are in  $U_i$, all the elements $z$ in $ y'$ are $C^E$-elements.
For the power-set,
for each $(\texttt{Pow}(C))^E$-element $y$, we require that all its elements $y' \in y$, which are not $U_1 \sqcup \ldots \sqcup U_k$-elements, are $C^E$-elements.  
We cannot simply define $(\texttt{Pow}(C))^E$ as $\texttt{Pow}(C^E)$, as it is necessary to keep the encodings of  $\forall R_i.C$ and $\texttt{Pow}(C)$ (both based on the set-theoretic power-set) independent of each other. 

Given an $\alc^\Omega$ knowledge base $K$, and a query $F$ (over the language of $K$), we need to define the encoding $K^E$ of $K$, 
and the encoding $F^E$ of the query $F$ in $\al^\Omega$. $K^E$ contains:

-  an inclusion axiom $C^E \sqcap \neg (U_1 \sqcup \ldots \sqcup U_k) \sqsubseteq D^E$, for each $C \sqsubseteq D \in K$;\footnote{Inclusion axioms are only required to hold on domain elements which are not $U_i$-elements.}

-  a membership axiom $C^E \in D^E$ for each  $C\in D$  in $K$; 

-  a membership axiom $a_i^E \in C^E$ for each $C(a_i)$  in $K$; 

-  membership axioms $F_{h,j}^i \in a_h^E$,  $a_j^E \in F_{h,j}^i$ and $F_{h,j}^i \in U_i$  for all $R_i(a_h,a_j)$  in $K$;

-  axioms $G_{C_h,C_j}^i \in C_h^E$,  $C_j^E \in G_{C_h,C_j}^i$ and $G_{C_h,C_j}^i \in U_i$  for all $R_i(C_h,C_j)$  in $K$.\footnote{To translate the assertion  $R_i(a_h,a_j)$, we need an element $u$ of $U_i$ such that  $u$ is an instance of $a_h^E$ and $a_j^E$ is an instance of $u$. We call such an element $F_{h,j}^i$. Similarly for role membership axioms $R_i(C_h,C_j)$.
A more direct encoding of role assertions would  be possible in the presence of nominals (i.e., if nominals were admitted in the languages of $\al^\Omega$ and $\alc^\Omega$)
as, for instance, $R_i(a_h,b_j)$ 
could be equivalently written as $a_h \in \exists R_i.\{a_j\}$.
However, this would require moving to a set theory with singleton operators.}

\noindent
The following additional axioms are also needed  in $K^E$:

 $A_i \sqsubseteq \neg (U_1 \sqcup \ldots \sqcup U_k)$, one for each concept name $A_i$ in $K$;
 
 $B_i \in \neg (U_1 \sqcup \ldots \sqcup U_k)$, one for each individual name $a_i$ in $K$;
 
 $C^E \in \neg (U_1 \sqcup \ldots \sqcup U_k) $, one for each $C$ 
 on the l.h.s  of a membership in $K$ or in the query; 
 
$ \neg (U_1 \sqcup \ldots \sqcup U_k) \sqsubseteq \texttt{Pow}(\neg (U_1 \sqcup \ldots \sqcup U_k)  \sqcup \texttt{Pow}(\neg (U_1 \sqcup \ldots \sqcup U_k) )) $\\
The first three axioms avoid that concept names $A_i$ and  $B_i$, and the concepts $C$ occurring on the l.h.s. of membership axioms, are interpreted as elements of $U_j$, for some $j$.
The last axiom enforces the property that: for each $z  \in \Delta \backslash (U_1 \sqcup \ldots \sqcup U_k)^I$, if $u \in z$, for an instance $u$ of some $U_i$, and $v \in u$, then 
$v \in  \Delta \backslash (U_1 \sqcup \ldots \sqcup U_k)^I$ (we call this axiom $Trans^2( \Delta \backslash (U_1 \sqcup \ldots \sqcup U_k)^I$).


For a query $F$ over the language of the knowledge base $K$, if $F$ is an inclusion $C  \sqsubseteq D$, its translation is $C^E  \sqsubseteq D^E$;  
if $F$ is an assertion $C(a_i)$, its translation is $a_i^E \in C^E$;
if $F$ is a membership axioms $C\in D$, its translation is $C^E \in D^E$.

\begin{example}
To see an example of the encoding above,
let us consider a variant of  knowledge base in Example \ref{exa:Eagle}.
Let $K= ({\cal T} ,{\cal A})$ where:
\\
${\cal T}= \{ \mathit{RedListSpecies \sqsubseteq \texttt{Pow}(CannotHunt), \ Eagle \sqsubseteq \forall hasMother. Eagle,  }$

$\;$ \ \ $ \mathit{RedListSpecies \sqsubseteq \forall hasScientificName. Name } \}$  
\\
${\cal A}=$ $\{\mathit{Eagle(harry), Eagle \in RedListSpecies, \; PolarCreature \sqcap Bear \in RedListSpecies,}$

$\;$ \ \  $\mathit{(PolarCreature \sqcap Bear, Eagle) \in moreEndangered}\} $

\noindent
By the translation above, we obtain the following $\al^\Omega$ knowledge base $K^E$:

$\mathit{RedListSpecies \sqsubseteq \texttt{Pow}(U_{hasMother} \sqcup U_{hasSciName} \sqcup CannotHunt)}$,

$ \mathit{Eagle \sqsubseteq \texttt{Pow}( \neg U_{hasMother}  \sqcup   \texttt{Pow}( Eagle)) }$,

$ \mathit{RedListSpecies \sqsubseteq \texttt{Pow}( \neg U_{hasSciName}  \sqcup  \texttt{Pow}(  Name ))}$,

$\mathit{B_{harry} \in Eagle}$, \
$\mathit{Eagle \in RedListSpecies}$, \
$\mathit{PolarCreature \sqcap Bear \in RedListSpecies}$, \

$\mathit{Eagle \in G_{PB,E}^{mE} }$, \ $\mathit{G_{PB,E}^{mE} \in PolarCreature \sqcap Bear}$,

$ \mathit{Eagle \sqsubseteq \neg (U_{hasMother}  \sqcup  U_{hasSciName} ) }$, \ and the same for the other concept names

$ \mathit{B_{harry} \in \neg (U_{hasMother}  \sqcup  U_{hasSciName} ) }$, 

$ \mathit{PolarCreature \sqcap Bear \in \neg (U_{hasMother}  \sqcup  U_{hasSciName} ) }$, 


\noindent
where $G_{PB,E}^{mE}$ is the concept name specifically introduced for encoding $\mathit{(PolarCreature}$
$\mathit{ \sqcap Bear, Eagle) \in}$ $\mathit{  moreEndangered}$.
The transitivity axiom $Trans^2$ is omitted.
\end{example}
We can prove  the soundness and completeness of the encoding of $\alc^\Omega$ into $\al^\Omega$.

\begin{proposition}[Soundness and Completeness of the encoding of $\alc^\Omega$ in $\al^\Omega$]\label{Prop:encoding}
 \begin{center}
$K \models_{\alc^\Omega} F  \mbox{ if and only if } K^E \models_{\al^\Omega} F^E$
\end{center}
\end{proposition}

\begin{proof}
($\Leftarrow$) The soundness is proved by contraposition.
Assume that $K \not \models_{\alc^\Omega} F $, then, there is a model $I=(\Delta, \cdot^I)$ of $K$ such that 
$F$ is falsified in $I$.

For the finite model property of $\alc^\Omega$, we can assume without loss of generality that the model $I$ is finite.
To build from $I$ a finite $\al^\Omega$ model $J=(\Delta', \cdot^J)$ of $K^E$ which falsifies $F^E$, 
we define a graph $G=(N , E)$ where: 
$N= \Delta \cup D_1 \cup \ldots \cup D_k$ 
and  $D_i = \{u^i_{s,t} \mid s,t \in \Delta \wedge (s,t) \in R_i^I \}$.
$E$ is defined as follows: 
\begin{center}
$E =  \bigcup_{i=1}^k \{(s, u^i_{s,t}), (u^i_{s,t},t) \mid s,t \in \Delta \wedge (s,t) \in R_i^I\}  \cup \{(s, t) \mid s,t \in \Delta \wedge t \in s\}$
\end{center}
We define an injection $\pi$ from the leaves of $N$, i.e. nodes without any successor, to  $\mathbb{A}$
and, for any given $ d\in N $, we define the following hyperset $ M(d) $: 
\begin{align*}
	M(d) & = \left\{\begin{array}{ll}
						\pi(d) & \mbox{ if } d  \mbox{ is a leaf of } N,\\
						\left\{M(d') \tc (d,d') \in E \right\} & \mbox{ otherwise. } 
					\end{array}\right.
\end{align*} 
The above definition uniquely determines hypersets in $ \HF^{1/2}(\mathbb{A}) $. 

Let $\Lambda=\{ M(d) \mid d \in N\}$, possibly extended  by duplicating M(d)'s to represent extensionally-equal (bisimilar) sets corresponding to pairwise distinct elements in $N$.
As a consequence, as in previous cases,  for $d,d' \in N$, $d=d'$ if and only if $M(d)=M(d')$, i.e., there are distinct sets in $\Lambda$ for pairwise distinct elements of $N$.

Observe that, by definition of $\Lambda$,  if $(s,t) \in R_i^I$, for $s,t \in \Delta$, then 
there is some $d \in D_i$, such that $M(d) \in M(s)$ and $M(t) \in M(d)$ (and, in particular, $d= u_{s,t}^i$); and vice-versa.

We define $ J= \langle \Delta', \cdot^J\rangle $ as follows:
\begin{quote}
	\ \ \ -  $\Delta'=\Lambda $;   
	  
	\ \ \ -  $A^{J}=\{ M(d) \mid d \in A^I\} $  for all $A \in N_C$,  in the language of $\alc^\Omega$;

	\ \ \ -  $B_i^J = M(a_i^I)= \pi(a_i^I)$, $i=1, \ldots, r$;
	         
	\ \ \ - $U_i^J= \{ M(u_{s,t}^i) \mid s,t \in \Delta \mbox{ and } (s,t) \in R_i^I\}$; 

	\ \ \ - $(F_{h,j}^i)^J = M(u_{a_h^I,a_j^I}^i) $;
	
	\ \ \ - $(G_{C_h,C_j}^i)^J = M(u_{s,t}^i)$, for  $s=C_h^I$ and $t= C_j^I$.
\end{quote}
\noindent
By construction,  $\Delta'$ is transitive set in a model $\emme$ of $\Omega$.
Notice that  $B_i^J = M(a_i^I) \in \mathbb{A}$, and hence $B_i^J$ has no elements.
Notice also that, in the definition of $(G_{C_h,C_j}^i)^J$, $s$ and $t$ are elements of $\Delta$ and $(s,t) \in R_i^I$, so that $u_{s,t}^i \in E$.
In fact, $s=C_h^I$ and $t= C_j^I$ and $R_i(C_h, C_j)$ is in ${\cal A}$.
Therefore, as $I$ satisfies  the ABox ${\cal A}$, $(C_h^I,C_j^I) \in R_i^I \subseteq \Delta \times \Delta$,
and $C_h^I,C_j^I \in \Delta$.
In the following, we let $M(D_i) = \{ M(u_{s,t}^i) \mid s,t \in \Delta \mbox{ and } (s,t) \in R_i^I\} $.

It can be shown by induction on the structural complexity of concepts, that, for all $d\in \Delta$,  for all concepts $C$ occurring in $K$ (or $F$):
\begin{align*}
d \in C^I \mbox{ if and only if } M(d) \in (C^E)^J, 
\end{align*}
which can be used to prove that $J$ is a model of $K^E$ that falsifies $F^E$, 
so that $K^E \not \models_{ \al^\Omega} F^E$.

\medskip

($\Rightarrow$) (Sketch) By contraposition,
assume that $K^E \not \models_{\al^\Omega} F^E $, then, there is an $\al^\Omega$ model $J=(\Delta, \cdot^J)$ of $K^E$ such that 
$F^E$ is falsified in $J$.

For the finite model property of $\al^\Omega$ (which is a fragment of $\alc^\Omega$), we can assume without loss of generality that the model $J$ is finite.
We build from $J$ an $\alc^\Omega$ model $I=(\Delta', \cdot^I)$ of $K$ which falsifies $F$, 
defining $\Delta'$ as a transitive set
in the universe $\HF^{1/2}(\mathbb{A})$ consisting of all the hereditarily finite rational hypersets
built from atoms in $\mathbb{A}=\{{\bf a_0}, {\bf a_1},\ldots  \}$.

We start from the graph $G=(N , E)$, with nodes $N= \Delta \backslash (U_1^J \cup \ldots \cup U_k^J )$, whose arcs are defined as follows:
$E=   \{ (d_1, d_2) \mid   d_1, d_2  \in N \wedge  d_2 \in d_1\}$.
$ G$ is finite. 
Observe that, for each $a_i$ in $K$, $B_i^J \in N$, by axiom $B_i \in \neg (U_1 \sqcup \ldots \sqcup U_k)$.
Similarly, for each $A_i$ in $K$, $A_i^J \subseteq N$.

We define an injection $\pi$ from the leaves of $N$ (i.e. nodes without any successor) plus the elements $B_1^I, \ldots, B_r^I\in N$ 
to  $\mathbb{A}$.
For any given $ d\in N $, we define the following hyperset $ M(d) $:  
\begin{align}\label{def_M(d)}
	M(d) & = \left\{\begin{array}{ll}
						\pi(d) & \mbox{ if } d  \mbox{ is a leaf of } N \mbox{ or } d=B_j^J \mbox{ for some $j$ },\\
						\left\{M(d') \tc (d,d') \in E \right\} & \mbox{ otherwise. } 
					\end{array}\right.
\end{align} 
The above definition uniquely determines hypersets in $ \HF^{1/2}(\mathbb{A}) $. 

$\Delta'=\{M(d) \tc d \in N\} $, possibly extended  by duplicating M(d)'s to represent extensionally-equal (bisimilar) sets corresponding to pairwise distinct elements in $N$.
We  complete the definition of $ I= \langle \Delta', \cdot^I\rangle $ as follows:

	
	-  $A^{I}=\{M(d)  \tc M(d) \in \Delta'  \wedge d \in A^J\}$,  for all $A \in N_C$;
	
	-  $R_i^{I}= \{(M(d),M(d')) \tc M(d), M(d') \in \Delta' \wedge \exists u \in U_i^J ( u \in d \wedge d' \in u) \}$, 
	
	   $ \mbox{  \ } $ for all roles $R_i$ occurring in $K$;  $\mbox{ }$ $R_i^{I}=\emptyset$ for all other roles $R\in N_R$;
		
	 - $a_i^I= M(B_i^J)= \pi(B_i^J)$  for all  named individuals  $a_i$ occurring in $K$;
	 
	       $ \mbox{  \ } $   $a^{I}=M(B_1^J)$ for all other $a \in N_I$.
	         	          
\noindent
By construction,  $\Delta'$ is a transitive set in a model $\emme$ of $\Omega$.
As $A^J \subseteq \Delta \backslash (U_1^J \cup \ldots \cup U_k^J )$, $A_i^I \subseteq \Delta'$.
To complete the proof it can be shown that, for all $M(d) \in \Delta'$, and $C$ in $K$ (or in $F$):
\begin{align*}
 M(d) \in C^I   \mbox{ if and only if }    d \in (C^E)^J 
\end{align*}
which can be used to prove that $J$ is a model of $K$ that falsifies $F$. 
\hfill $\Box$
\end{proof}

Combining the above encoding  and the set-theoretic translation for $\al^\Omega$ of Section \ref{sec:tr_ALC^Omega-R to sets},
we obtain a set-theoretic translation for $\alc^\Omega$.

Let $R_1,\ldots,R_k$  and  $A_1, \ldots, A_n$ be, respectively, the roles and  the concept names occurring in the knowledge base  $K= ({\cal T}, {\cal A})$ (or in the query).
Given a concept $C$ of $\alc^\Omega$, built from the concept, role  and individual names in $K$, 
its set-theoretic translation $(C^E)^S$ is a set-theoretic term $C^*(x,y_1,\ldots, y_k, x_1,$ $\ldots, x_{n+m})$, 
where we let $U_i^S=y_i$, and we let the variables $x_{n+1}, \ldots, x_{n+m}$ to be the set-theoretic translation of the additional concept names ($B_j$, $F_{h,j}^i$  and 
$G_{C_h,C_j}^i$) introduced to encode assertions. 
$C^*$ is defined inductively as follows:
\medskip

 $\;$ \ \ \ \ \ \   $\top^*=x$;   \ \ \ \ \ \ \ \ \ \ \ \ \ \ \ \ \ \ \ \ \ \ \ \ \ \ \ \ \ \ \  \ \ \ \ \  \ \ \ \ \   \ \ \ \ \  \ \ \ \ \ \ \   \ \ \ \ \   $\bot^*=\emptyset$;
   
$\;$ \ \ \ \ \ \    $A_i^*= x_i$ , for $A_i $  in $K$;     \ \ \ \ \ \ \  \ \ \ \ \ \   \ \ \ \ \  \ \ \ \ \  \ \ \ \ \ \ \  \ \ \ \ \   $(\neg C)^*=x \backslash C^S$;
   
$\;$ \ \ \ \ \ \    $(C \sqcap D)^*=C^S \cap D^S$;  \ \ \ \ \ \ \  \ \ \ \ \   \ \ \ \ \   \ \ \ \ \  \ \ \ \ \ \ \  \ \ \ \ \   $(C \sqcup D)^*=C^S \cup D^S$;
  
$\;$ \ \ \ \ \ \   $(\forall R_i.C)^*= Pow((x  \backslash y_i ) \cup Pow(C^*))$; 
       \ \ \ \ \ \ \  \ \  $ \texttt{Pow}(C)^*= Pow( y_1 \cup \ldots \cup y_k  \cup C^*)$.
 
\medskip

The translation of an $\alc^\Omega$ knowledge base $K$ can be defined accordingly, exploiting the encoding $E$ and $S$ of the KB. %
In particular,  let $\mathit{TBox}^*_{\cal T}$ and $\mathit{ABox}^*_{\cal A}$ be the set-theoretic translation of ${\cal T}$ and ${\cal A}$, respectively.
Observe that $Trans^2((x \backslash ( y_1 \cup \ldots \cup y_k)))$ is in $ \mathit{TBox}^*_{\cal T}$.
Also, for each $C_1 \sqsubseteq C_2  \in {\cal T}$, $ C_1^* \cap (x \backslash ( y_1 \cup \ldots \cup y_k)) \subseteq C_2^* $ is in $ \mathit{TBox}^*_{\cal T}$.

 $\mathit{ABox}^*_{\cal A}$ contains $C_1^* \in C_2^* \cap  (x  \backslash ( y_1 \cup \ldots \cup y_k))$, for each $C_1 \in C_2  \in {\cal T}$
(from axioms $C_1^E \in C_2^E$ and  $C^E \in \neg (U_1 \sqcup \ldots \sqcup U_k) $ in $K^E$ )
 and, in addition, 
$B_i^*  \in  (x  \backslash ( y_1 \cup \ldots \cup y_k))$, for each individual name $a_i$ occurring in ${\cal A}$.

A set-theoretic translation for subsumption in $\alc^\Omega$ follows from the encoding $E$ and the set-theoretic translation $S$ for $\al^\Omega$ in Section \ref{sec:tr_ALC^Omega-R to sets} (see Proposition \ref{Prop:encoding}):

\begin{corollary}\label{corollary:tr_ALC^Omega}
$K \models_{\alc^\Omega} C \sqsubseteq D  \text{ if and only if }$ 

$\Omega \models  \forall x, \forall y_1, \ldots, \forall y_k (  Trans(x)  
\rightarrow $

$\;$ \ \ \ \ \ \ \ \ \ \ \ \ \ \ \ \  
$ \forall x_1, \ldots, \forall x_{n+m} ( \bigwedge \mathit{ABox}^*_{\cal A} \wedge   \bigwedge \mathit{TBox}^*_{\cal T} $
$\rightarrow C^* \cap  (x \backslash ( y_1 \cup \ldots \cup y_k)) \subseteq D^*))$
\end{corollary}

\noindent 
Rewriting $\mathit{TBox}^*_{\cal T} $ as $Trans^2((x \backslash ( y_1 \cup \ldots \cup y_k))) \wedge \mathit{TBox}^{*-}_{\cal T} $  and observing that we can factorise out $Trans^2((x \backslash ( y_1 \cup \ldots \cup y_k)))$, we can put: 

$\Omega \models  \forall x, \forall y_1, \ldots, \forall y_k (  Trans(x)  \wedge Trans^2((x \backslash ( y_1 \cup \ldots \cup y_k))) 
\rightarrow $

$\;$ \ \ \ \ \ \ \ \ \ \ \ \ \ \ \ \  
$ \forall x_1, \ldots, \forall x_{n+m} ( \bigwedge \mathit{ABox}^{*}_{\cal A} \wedge   \bigwedge \mathit{TBox}^{*-}_{\cal T} $
$\rightarrow C^* \cap  (x \backslash ( y_1 \cup \ldots \cup y_k)) \subseteq D^*))$

\noindent
which makes it more evident that this set-theoretic translation of $\alc^\Omega$ is a generalization of the translations given in Section  \ref{sec:tr_ALC to sets} and in  Section \ref{sec:tr_ALC^Omega-R to sets}.

When the power-set construct does not occur in the KB, and the language is restricted to the language of $\alc$,  it corresponds to the set-theoretic translation of $\alc$ in Section \ref{sec:tr_ALC to sets}.
Here, the set $x \backslash ( y_1 \cup \ldots \cup y_k)$ plays the role of $x$ in $\alc$ translation in Proposition  \ref{set-th-translation-alc}.  
Condition $Trans^2(x \backslash ( y_1 \cup \ldots \cup y_k))$ correspondes to condition $Trans^2(x)$. The inclusions $C_1^* \cap  (x \backslash ( y_1 \cup \ldots \cup y_k)) \subseteq C_2^*))$ in $ TBox^*_{\cal T}$ (and in the query) correspond to the inclusions
$C_1^* \cap  x \subseteq C_2^*)) \in TBox_{\cal T}$ (and in the query).
Condition $Trans(x)$ is useless (but harmless) in this case.

When there are no roles, no assertions and no universal and existential restrictions, the set-theoretic variables  $y_1, \ldots, y_n$  are useless.
Let us consider the case when,  in the translation above, the interpretation of $y_1, \ldots, y_n$ is the empty set.
In such a case, $x \backslash ( y_1 \cup \ldots \cup y_k) = x$ and the the set-theoretic encoding above collapses to the set-theoretic encoding of $\al^\Omega$ in Proposition \ref{set-th-translation}. In particular, condition $Trans^2(x \backslash ( y_1 \cup \ldots \cup y_k))$ becomes $Trans^2(x)$, which trivially follows from $Trans(x)$.

The correspondence above provides a set-theoretic translation for $\alc$, which is slightly different w.r.t. the translation in Section \ref{sec:tr_ALC to sets}, directly obtained from the application of the result for  normal polymodal logics in \cite{DAgostino1995} and from Schild's characterization of $\alc$ as a polymodal logic \cite{Schild91}.
The reason is that, in the encoding $E$ above, the concept names $U_i$, playing the role of the sets $y_i$ in the set-theoretic translation in Section \ref{sec:tr_ALC to sets}, are in the language of $\al^\Omega$ and hence are interpreted in $\Delta$ as all other concept names.

\section{Discussion} 
We  consider here some consequences of the above results in view of possible extensions of the correspondence  we introduced
 to deal with further DL constructs or set-theoretic operators.

First of all, observe that the complementary problem to subsumption 
 in $\alc^\Omega$ is satisfiability of a concept  $C$ 
with respect to a general knowledge base $K$  (see Section \ref{sec:ALC}). 
By the result in Corollary \ref{corollary:tr_ALC^Omega}, this problem  corresponds to the satisfiability of a formula in the existential fragment of $\Omega$,
i.e. the satisfiability of a formula of $\Omega$ of the form $\exists x, y_1, \ldots, y_k,  x_1, \ldots, x_{n+m}\; \phi$, where $\phi$ is unquantified in the basic language of $ \Omega $.
In fact, we can reformulate Corollary \ref{corollary:tr_ALC^Omega} as follows:

{\em  $K \not \models_{\alc^\Omega} C \sqsubseteq D $ 
($\text{ i.e., }$ $C \sqcap \neg D$ is satisfiable in $\alc^\Omega$ 
with respect to  $K$) 

$ \text{ iff }$ there is a model of $\Omega$ satisfying the formula:}

$\;$ \ \ \ \ $ \exists x, \exists y_1, \ldots, \exists y_k (  Trans(x)  
\wedge $

$\;$ \ \ \ \ \ \ \ \ \ \ 
$ \exists x_1, \ldots, \exists x_{n+m} ( \bigwedge \mathit{ABox}^*_{\cal A} \wedge   \bigwedge \mathit{TBox}^*_{\cal T} $
$\wedge \neg (C^* \cap  (x \backslash ( y_1 \cup \ldots \cup y_k)) \subseteq D^*)))$

\noindent
where the quantifiers $ \exists x_1, \ldots, \exists x_{n+m} $ can then be moved in front of  the first parenthesis, thus giving a formula in the existential fragment of (the language of) $\Omega$.

The problem of deciding the satisfiability of  existential formulae of the theory $\Omega$ (without extensionality and well-foundedness)
has not been studied so far, and our decidability result for subsumption in $\alc^\Omega$ comes from the translation of $\alc^\Omega$ into the description logic $\alcio$ (Propositions \ref{Prop:Soundness-translation} and \ref{Prop:Completeness-translation} in Section \ref{sec:ALC^Omega to ALCIO}). 
%
However, the satisfiability of  existential formulae with power-set relative to a set theory assuming extensionality and well-foundedness, has been proved to be decidable  by Cantone et al. in  \cite{CanFerSch85}.
As a consequence, 
the {\em same} set-theoretic translation considered above brings us naturally to 
a well-founded and extensional variant of  $\alc^\Omega$.
Let us elaborate on this point.

Start from the class of formulae ``Multilevel syllogistic extended by the powerset operator" ({\em MLS+Pow} for short) whose decidability has been studied  in \cite{CanFerSch85} under the assumptions of extensionality and well-foundedness of the underlying set theory. {\em MLS+Pow} consists of purely existential formulas with a matrix in the language of $\Omega$. Introduce (as we have done in Section \ref{sec:tr_ALC^Omega-R to sets}   with $\al^\Omega$)  a very simple description logic, $\al_{we}^\Omega$, which is (basically) a fragment of {\em MLS+Pow} with extensionality and well-foundedness.
$\al_{we}^\Omega$ has the same syntax of $\al^\Omega$ and,
from the semantic point of view,  models of $\al_{we}^\Omega$ can be defined  as  models of $\al^\Omega$, with the additional requirements 
of extensionality and well-foundedness.  The  decidability proof in \cite{CanFerSch85} also provides a finite-model result for $\al_{we}^\Omega$.

A description logic  $\alc_{we}^\Omega$, extending $\alc$ with well-founded/extensional sets and with the power-set operator, under the assumption of extensionality,
can then be defined and translated into its fragment $\al_{we}^\Omega$, along the lines of our translation of  $\alc^\Omega$ into $\al^\Omega$.
Indeed, the encoding in Section  \ref{sec:encoding} of $\alc$ roles by means of the membership operator (using the power-set to capture the universal restriction)
is still possible in the case of well-founded sets, as $\alc$ has the ``tree-model property" (as the polymodal logic $K_m$ \cite{BlackburnRV01}).  
From the semantic point of view, 
by well-foundednness, all circular membership relationships among concepts are ruled out;
by extensionality, any two concepts in a model of $\alc_{we}^\Omega$ having the same elements, and such that the same domain elements are accessible through the relations associated with roles $R_i$, have to be considered equal. 
For instance, if we have a knowledge base containing the axioms $Eagle\equiv Aquila$ and $Eagle \in RedListSpecies$, by extensionality we can conclude that $Aquila \in RedListSpecies$. Instead, in $\alc^\Omega$ (without extensionality) the concepts  $Eagle$ and $Aquila$ may be interpreted as different sets  in the models of the knowledge base (i.e., $Eagle^I \neq Aquila^I$), although they have the same elements. 
In this case,  as we have seen in Section \ref{sec:DL+Omega}, $Aquila \in RedListSpecies$ does not follow from the knowledge base.

Further decidable existential fragments of set-theory have been studied, such as the fragment with power-set and singleton operators  \cite{Can91,CanOmoUrs02},
which opens the way to a set-theoretic definition of other decidable extensions of $\alc$ with well-founded sets. 
A natural question arising is whether these description logics with power-set, well-founded sets and extensionality, can be translated as well into standard DLs
and, in particular, whether the extensionality and well-foundedness assumptions can be captured in standard DLs. We leave this investigation 
for future work.

Concerning the expressivity of $\alc^\Omega$ with respect to standard DLs, the fact that it has the finite model property already makes it evident that it cannot capture combinations of constructs of expressive DLs which do not satisfy this property.
On the one hand, one can consider the problem of identifying a description logic (if any) having the same expressivity as $\alc^\Omega$ and, on the other hand, one can face the problem of extending other description logics (including expressive ones) with power-set  and concept membership.

For the first point, 
it is not likely that the logic $\alcio$ can be translated into $\alc^\Omega$. Indeed, while inverse roles can be encoded set-theoretically as in Section \ref{sec:tr_ALC to sets}, their 
encoding does not seem to be easily turned into a prenex universal formula so that, with their addition,  concept satisfiability (w.r.t. a knowledge base) seems to fall outside the existential fragment of $\Omega$.

As regards extending expressive DLs with power-set and concept membership, we observe that 
their encoding using nominals and an inverse role of $\alcio$, can be exploited in any DL extending $\alc$ and including the above mentioned constructs. The proof of soundness of the translation  of $\alc^\Omega$ in $\alcio$  (Proposition \ref{Prop:Soundness-translation}) indeed generalizes to other DLs, when extended with the power-set construct and concept membership in a similar way. Instead, alternative techniques would be needed for proving completeness of the translation 
for such logics, as the proof of Proposition \ref{Prop:Completeness-translation} exploits the finite model property of $\alc^\Omega$ which, in general, is not a property of expressive DLs.

There are other useful constructs which could be borrowed from set-theory and added to $\alc^\Omega$.
One of them is the {\em unary union}, $\bigcup  C$, namely, the union of all the subsets of concept $C$. This construct could be introduced in $\alc^\Omega$
with the semantic condition $(\bigcup  C )^I =\{x \in y \mid y \in C^I \}$.
Let us consider again Example \ref{exa:readingGroup}.
\begin{example} 
We may want to introduce an association $\mathit{ACME}$ and state that all the members of   $\mathit{ACME}$
participate to the $\mathit{SummerMeeting}$ (in some group).
Notice that, using unary union, we can represent the set of all the participants to the $\mathit{SummerMeeting}$ as $\bigcup \mathit{SummerMeeting}$,
We can state that all the members of $\mathit{ACME}$ participate to the summer meeting by the inclusion: 
$$\mathit{ACME \sqsubseteq \bigcup SummerMeeting}$$
and that  the members of $\mathit{ACME}$ are all and the only participants to the summer meeting by $\mathit{ACME \equiv \bigcup SummerMeeting}$.
\end{example}
$\alc^\Omega$ can be easily extended with the construct of unary union $ \bigcup C$.
On the one hand, this construct has a natural translation into $\alcio$, by introducing, for each $\bigcup C$ occurring in the knowledge base $K$, a new concept name $ U_C$ together with the axiom $ U_C \equiv \exists e^-. C$, and replacing all the occurrences of $ \bigcup C$ in $K$ with $U_C$.
%
While the proof of soundness in Proposition \ref{Prop:Soundness-translation} extends to this case,
alternative techniques would be needed for proving completeness of the translation for the case with unary union, as the proof of Proposition \ref{Prop:Completeness-translation} exploits the finite model property of $\mathcal{ALC}^{\Omega}$ which we do not expect to hold in this case. 
%
%
%
Adapting the encoding of $\alc^\Omega$ into $\al^\Omega$ is not immediate and  it might require a change in the encoding of the power-set concept as well. We leave the study of this encoding for future investigation.

As we will see in the next section, a weakness of $\alc^\Omega$ with respect to other extensions of description logics dealing with metamodeling, is that 
$\alc^\Omega$ does not allow roles as instances of concepts.
For instance, one could want to define a concept $\mathit{Relatives}$ including the roles $\mathit{hasParent, hasCusin, hasSibling,}$ etc. and,
in the formalisms admitting roles as elements of concepts, such as those in \cite{Motik05,Kubincova2016}, one can indeed state that $\mathit{hasParent \in Relatives}$, $\mathit{hasCusin \in Relatives}$, $\mathit{hasSibling \in Relatives}$.
This is not possible in $\alc^\Omega$. 
The encoding of $\alc^\Omega$ into $\al^\Omega$ in Section  \ref{sec:encoding}, that associates
 a concept $U_i$ with each role $R_i$, might suggest a possible translation of membership axioms $R_i \in C$   into $\al^\Omega$ as $U_i \in C$.
Indeed, each element  $u \in U_i$ represents a set of pairs $(y, z)$ of domain elements such that  $z \in u \in y$ and, therefore, $(y, z) \in R_i$.
The feasibility of such 
an extension and the study of a possible translation of such membership relations into standard description logics are left for future work.

\section{Related work} 

The power-set construct allows to capture in a very natural way the interactions between concepts and metaconcepts, adding to the language of $\alc$ the expressivity of metamodelling.  
The issue of metamodelling has been analysed by Motik \cite{Motik05}, who proved that metamodelling in $\alc$-Full is already undecidable
due to the free mixing of logical and metalogical symbols. Two decidable semantics, a contextual $\pi$ semantics 
and a Hilog $\nu$-semantics, are introduced in \cite{Motik05} for a language extending $\shoiq$ with metamodelling, where concept names, role names and individual names are not disjoint.  The possibility of using the same name in different contexts is introduced in OWL 1.1  and then in OWL 2 through {\em punning}\footnote{https://www.w3.org/2007/OWL/wiki/Punning}. 
As a difference, in this paper, we consider concept names, role names and individual names to be disjoint,  we allow concepts (and not only concept names) to be instances of other concepts, by membership axioms, while we do not allow role names as instances.

As in \cite{Motik05}, DeGiacomo et al. \cite{DeGiacomo2011}  and Homola et al. \cite{HomolaDL14} employ an Hilog-style semantics
to define $\mathit{Hi(\shiq)}$ and $\mathit{{\cal TH}(\sroiq)}$, respectively.
While \cite{Motik05} and \cite{DeGiacomo2011} define untyped higher-order languages which, as $\alc^\Omega$, allow a concept to be an instance of itself,
 \cite{HomolaDL14} defines a typed higher-order extension of $\sroiq$ allowing for a hierarchy of concepts,
where concept names of order $t$ can only occur as instances of concepts of order $t+1$.
In $\mathit{{\cal TH}(\sroiq)}$ 
 \cite{HomolaDL14} 
there is a strict separation between concepts  and roles (as in $\alc^\Omega$) and decidability is proved by a polynomial 
first-order reduction into $\sroiq$, which generalizes the reduction in \cite{GlimmISWC2010} to an arbitrary number of orders. 
The translation in \cite{HomolaDL14}  introduces  axioms $\mathit{A' \equiv \exists instanceOf.\{c_{A'}\}}$, for each concept name $A'$,
axioms which are quite similar to our axiom (\ref{axiom:e_c}), that we need for the concepts $C$ occurring in the knowledge base on the left hand side of membership axioms. 

In $\mathit{Hi(\shiq)}$ \cite{DeGiacomo2011}, complex concept and role expressions can occur as instances of other concepts as in   $\alc^\Omega$. 
A polynomial translation of $\mathit{Hi(\shiq)}$ into $\shiq$  is defined and a study of the complexity of higher-order query answering is provided.

Kubincova et al. in \cite{KubincovaDL15} propose a Hylog-style semantics, dropping the ordering requirement in \cite{HomolaDL14} and allowing the $\mathit{instanceOf}$ role, 
with a fixed interpretation, to be used in axioms as any other role.  
The interpretation of role $\mathit{instanceOf}$ does not correspond exactly to the interpretation of  $e^-$ in our translation, 
as we do not introduce  axiom  (\ref{axiom:e_c}) for all the concept names in $N_C$, while we introduce it for all the
concepts occurring as instances in some membership axiom. 
In \cite{Kubincova2016} Kubincova et al. define the description logic  ${\cal HIR(\sroiq)}$,  an extension of $\sroiq$ with an HiLog-style semantics, which maintains basic separation between individuals, concepts, and roles, but allows for meta concepts and meta roles which are promiscuous (they can classify/relate any entities). 
The logic features a fixedly interpreted instanceOf role, modeling the instantiation relation.

Pan et al. in \cite{Pan2005} and Motz et al. in  \cite{Motz2015}  define extensions of OWL DL and of $\shiq$ (respectively), 
based on semantics interpreting concepts as well-founded sets.  
In particular, \cite{Motz2015} adds to $\shiq$ meta-modelling axioms equating individuals to concepts, without requiring that 
the instances of a concept need to stay in the same layer, and develop a  tableau algorithm as an extension of the one for $\shiq$. 

In  \cite{Gu2016} Gu introduces the language Hi(Horn-SROIQ),  an extension of Horn-SROIQ which allows classes and roles to be used as individuals
based on the $\nu$-semantics \cite{Motik05}.
$\nu$-satisfiability and conjunctive query answering are shown to be reducible to the corresponding problems in Horn-SROIQ.

Badea in \cite{Badea1997} first suggested a way of representing the power-set 
in a reified ${\cal ALCO}_\in$, using the universal restriction and two roles $\in$ and  $\ni$. 
 \cite{Badea1997} does not consider an higher-order semantics, 
but interprets ``quantified concept variables as ranging over (explicitly given) reified individuals",
and develops a calculus for checking consistency in reified ${\cal ALCO}_\in$.
As a difference,  here we show that a semantics quantifying over a transitive set in the universe of an $\Omega$-model can be mapped to standard DLs.
\normalcolor

A set-theoretic approach in DLs has been adopted by Cantone et al. in \cite{CantoneAS17} 
for determining the decidability of  higher order conjunctive  query answering in the description logic ${\cal DL}^{4,\times}_D$ 
(where concept and role variables may occur in queries),
as well as for developing a tableau based procedure 
for dealing with several well-known ABox reasoning tasks.

\section{Conclusions} 

In this paper we have shown that the similarities between Description Logics and Set Theory can be exploited to introduce in DLs  the new power-set construct and to allow for (possibly circular) membership relationships among arbitrary concepts.
We started from the description logic $\alc^\Omega$---combining $\alc$ with the set theory $\Omega$---whose interpretation domains are fragments of the domains of $\Omega$-models. 
$\alc^\Omega$ allows membership axioms among concepts as well as the power-set construct 
which, apart from \cite{Badea1997}, has not been considered for description logics before.
We show that an $\alc^\Omega$ knowledge base can be  polynomially translated into an $\alcio$ knowledge base,
providing, besides decidability, an \textsc{ExpTime} upper bound for satisfiability in $\alc^\Omega$.
We also develop a set-theoretic translation for the description logic $\alc^\Omega$  
into the set theory $\Omega$ exploiting a technique, originally proposed in \cite{DAgostino1995}, 
for translating normal modal and polymodal logics into $\Omega$.
The translation has been defined step by step, first defining a set-theoretic translation for $\alc$ with empty ABox, then for $\al^\Omega$,
the fragment of $\alc^\Omega$ without role names and individual names and, finally, providing an encoding of $\alc^\Omega$ into $\al^\Omega$.
The paper extends the preliminary results in \cite{ICTCS_2018}  and \cite{Jelia_2019}, which do not consider a set-theoretic encoding of role assertions and role membership axioms, and  exploit a slightly stronger semantics.

The set-theoretic translation, on the one hand, clarifies the nature of the power-set concept  (which indeed corresponds to the set-theoretic power-set,
provided the valuation of inclusions is restricted to the set corresponding to the domain $\Delta$) 
and, on the other hand, shows that the fragment of $\al^\Omega$ without roles and individual names 
is as expressive as  $\alc^\Omega$. 
The correspondence among fragments of  set theory and description logics may open to the possibility of transferring proof methods or decidability results across the two formalisms.

The set-theoretic translation of $\alc^\Omega$ can be extended to constructs of more expressive DLs, and this approach suggests a way to incorporate the power-set construct in more expressive DLs. As the proof techniques used in this paper exploit the finite model property of $\alc^\Omega$,
alternative techniques will be needed to deal with more expressive DLs. 
 Other possible directions of  future investigation 
are, as mentioned above, the study of variants of $\alc^\Omega$ semantics with well-foundedness and extensionality (and, specifically, of their translation to DLs)
and the treatment of roles as individuals, which has not been considered as an option in $\alc^\Omega$.

\normalcolor

\medskip
{\bf Acknowledgement:} 
This research is partially supported  by INDAM-GNCS 
Project 2019  "METALLIC \#2: METodi di prova per il ragionamento Automatico per Logiche non-cLassIChe"

\bibliographystyle{plain}

\bigskip

\begin{appendix}

\noindent
{\bf \Large Appendix}


\bigskip

\noindent
{\bf Proposition \ref{Prop:Soundness-translation} (Soundness of the translation)} 
{\em
The translation of an $\alc^\Omega$  knowledge base $K= ({\cal T} ,{\cal A})$  into $\alcio$ is sound,  that is, 
for any query $F$:
\begin{align*}
	K^T \models_{\alcio} F^T & \Rightarrow  K \models_{\alc^\Omega} F.
\end{align*}
}
\begin{proof}
By contraposition, assume  $K \not \models_{\alc^\Omega} F$ and let 
$I=\langle \Delta, \cdot^I \rangle$ be a model of $K$ in ${\alc^\Omega} $ that falsifies $F$.
$\Delta$ is a transitive set living in a model of $\Omega$ with universe ${\cal U}$.

We build an $\alcio$ interpretation $I'=\langle \Delta', \cdot^{I'} \rangle$, which is going to be a model of $K^T$ falsifying $F$ in $\alcio$, by letting:

-  $\Delta'=\Delta$;  

- for all $B \in N_C$, $B^{I'}=B^{I}$;    

-  for all roles $R \in N_R$, $R^{I'}= R^I$;

-  for all $x,y \in \Delta'$, $(x,y) \in e^{I'}$ if and only if $y \in x$;

-  for all (standard) individual name $a \in N_I$, $a^{I'}=a^I \in \mathbb{A} \cap \Delta$;

-  for all $e_C  \in N_I$, $e_C ^{I'}=C^I.  $ \\
The interpretation $I'$ is well defined. First, the interpretation $B^{I'}$ of a named concept $B$ is a subset of $\Delta'$ as expected. 
In fact,  for each $x\in B^{I'}$,  $x\in B^I \subseteq \Delta = \Delta'$.
Also, $a^{I'}=a^I \in \Delta= \Delta'$.
It is easy to see that the interpretation of constant $e_C$, $e_C^{I'}$ is in $\Delta'$. In fact, 
as the named individual $e_C$ has been added by the translation to the language of $K^T$, there must be some membership axiom $C \in D$ (or $(C,D) \in R$) in $K$, for some $D$ (respectively, for some $D$ and $ R$).
Considering the case that axiom $C \in D$ is in $K$, as $I$ is a model of $K$, $I$ satisfies $C \in D$, so that $C^I \in D^I$ must hold.
However, as $D^I \subseteq \Delta$, it must be $C^I \in \Delta$.
Hence, by construction, $e_C ^{I'} = C^I \in \Delta'=\Delta$.
In case $(C,D) \in R$, it must hold that $(C^I,D^I) \in R^I$. 
As $R^I \subseteq \Delta \times \Delta$, then $C^I, D^I \in \Delta$.
In particular, $e_C ^{I'} = C^I \in \Delta'=\Delta$.

We can prove by induction on the structural complexity of the concepts that, 
for all $x \in \Delta'$,
\begin{equation} \label{eq}
 x \in (C^T)^{I'} \mbox{ if and only if } x \in C^I
 \end{equation}
 
 For the base case, the property above holds  for $C= \top$ and  $C = \bot$,
as $ \top^T=\top$ and  $\bot^T=\bot $, and it also holds  by construction for all concept names $ B \in N_C$. 
 
 For the inductive step, 
 let $C= C_1 \sqcap C_2$ and let $x \in ((C_1 \sqcap C_2)^T) ^{I'} =(C_1^T \sqcap C_2^T) ^{I'}$,  for some $x \in \Delta'$. 
 As $I'$ is an $\alcio$ interpretation, $x \in (C_1^T) ^{I'} $ and $x \in (C_2^T) ^{I'} $ and since by induction (\ref{eq}) holds for concepts $C_1$ and $C_2$, we have $ x \in C_1^I $  and  $x \in C_2^I $.
Therefore, $ x \in (C_1^I  \cap  C_2^I)$ and, 
by definition of an ${\alc^\Omega}$ interpretation, $x \in (C_1 \sqcap C_2)^I$.
It is easy to see that  the vice-versa also holds, i.e., if $x \in (C_1 \sqcap C_2)^I$ 
then $x \in ((C_1 \sqcap C_2)^T) ^{I'} $.

For the case $C = \texttt{Pow}(D)$, let $x \in ((\texttt{Pow}(D))^T) ^{I'} $ $= (\forall e. D^T)^{I'}$,  for some $x \in \Delta'$.
As $I'$ is an $\alcio$ interpretation, for all $y \in \Delta'$, if $(x,y) \in e^{I'}$ then $y \in (D^T)^{I'}$.
By construction of $I'$, $(x,y) \in e^{I'}$ if and only if $y \in x $ and,
by inductive hypothesis, $y \in (D^T)^{I'}$ if and only if $y \in D^{I}$. Hence, 
for all $y\in \Delta$ such that $y \in x$, $y \in D^{I}$.
Therefore, $x \cap \Delta \subseteq D^I$. As $x \in \Delta'=\Delta$, and $\Delta$ is transitive, then  $x \cap \Delta =x$.
Therefore, $x \subseteq D^I$, and  $x \in \mathit{Pow}(D^I)$. As $x \in \Delta$,  $x \in  \mathit{Pow}(D^I) \cap \Delta=(\texttt{Pow}(D))^{I} $. 
The vice-versa can be proved similarly.  

%

Let us consider the case $C = \exists R.D$. Let $ x \in ((\exists R. D)^T)^{I'}$ for some $x \in \Delta'$.
As $ x \in (\exists R. D^T)^{I'}$ and $I'$ is an $\alcio$ interpretation, there is a $y \in \Delta'$ such that $(x,y) \in R^{I'}$ and $y \in (D^T)^{I'}$. 
By inductive hypothesis, $y \in D^{I} $. 
Furthermore, by construction of $I'$, it must be that $(x,y) \in R^{I}$ and  $x, y \in \Delta$.
Hence, $ x \in (\exists R. D)^{I}$. The vice-versa can be proved similarly as well as 
 all the other cases for the concept $C$.

\medskip

Using  (\ref{eq}) we can now check that all axioms and assertions in $K^T$ are satisfied in $I'$.

For an inclusion axiom $C^T \sqsubseteq D^T\in {\cal T}^T$, the corresponding inclusion axiom  $C \sqsubseteq D$ is in ${\cal T}$.
If $x \in (C^T)^{I'}$ for some $x \in \Delta'$, by  (\ref{eq}) $x \in C^{I}$  and, by the inclusion $C \sqsubseteq D \in {\cal T}$, $x \in D^{I}$.
Hence, again by (\ref{eq}), $x \in (D^T)^{I'}$.

Each assertion $D^T(a) \in {\cal A}^T$, is obtained from the translation of the assertion $D(a) \in{\cal A}$.
From the fact that $D(a)$ is satisfied by $I$, i.e. $a^I \in D^I$, given property (\ref{eq}),
it follows that $a^{I'} =a^I \in (D^T)^{I'}$.

For each assertion $D^T(e_C ) \in {\cal A}^T$ obtained from the translation of a membership axiom $C \in D$,
from the fact that $I$ is a model of $K$, we know that $C^I \in D^I$ holds.
By construction,  $e_C ^{I'}= C^I$ and we have seen that $C^I \in \Delta= \Delta'$.
From $C^I \in D^I$, it follows that $e_C ^{I'} \in D^I$ and,
by property (\ref{eq}), $e_C ^{I'} \in (D^T)^{I'}$.

%
For each assertion $R(e_C,e_D) \in {\cal A}^T$ obtained from the translation of a role membership axiom $(C,D) \in R$,
from the fact that $I$ is a model of $K$, 
we know that $(C^I  ,D^I) \in R^I$ and  $C^I, D^I \in \Delta=\Delta'$  hold.
We want to show that $(e_C ^{I'},e_D ^{I'}) \in R^{I'}$. As, by construction, $e_C ^{I'}= C^I \in \Delta= \Delta'$ 
and $e_D ^{I'}= D^I \in \Delta= \Delta'$
from $(C^I, D^I ) \in R^I$, it follows that $(e_C ^{I'},e_D ^{I'}) \in R^I$.
By the definition of role interpretation in $I'$, $(e_C ^{I'},e_D ^{I'})\in R^{I'}$. 

For each assertion $(\neg \exists e. \top)(a)$, for $a \in N_I$, it is easy to see that $a^{I'} \not \in (\exists e. \top)^{I'}$.
As $a^{I'} = a^I \in \mathbb{A}$ and an element of $\mathbb{A}$ in $\Delta$ is interpreted as an empty set, there is no $y$ such that $y \in a^I $.
Hence, by definition of $e^{I'}$ in the model $I'$, there is no $y \in \Delta'$ such that $(a^{I'},y) \in e^{I'}$.

We still need to show that  axiom $C^T \equiv \exists e^-.\{e_C \}$ is satisfied in $I'$ 
for all the concepts $C$ occurring in $K$ 
on the l.h.s. of membership axioms.
Let $x \in (C^T)^{I'}$. By property (\ref{eq}),  $x \in C^I$
and, by construction, $e_C ^{I'} =  C^I \in \Delta$.
We want to show that $x \in (\exists e^-.\{e_C \})^{I'}$, i.e. that
$(e_C ^{I'}, x) \in e^{I'}$.
By the definition of $e^{I'}$ in the $\alcio$  interpretation $I'$,
$(e_C ^{I'}, x) \in e^{I'}$ if and only if $x \in e_C ^{I'} $.
But $x \in e_C ^{I'} $ immediately follows from the previous conclusions that  $x \in C^I$, as $e_C ^{I'} =  C^I $ by construction.
The vice-versa can be proved similarly. 

To conclude the proof, it can be easily shown that, for any query $F$, $F^T$ is satisfied in $I'$ if and only if $F$ is satisfied in $I$.
\hfill $\Box$
\end{proof}

\bigskip


\noindent
{\bf Proposition \ref{Prop:Completeness-translation} (Completeness of the translation)} 
{\em
The translation of an $\alc^\Omega$  knowledge base $K= ({\cal T} ,{\cal A})$  into $\alcio$ is complete,  that is, 
for any query $F$:
\begin{align*}
	K \models_{\alc^\Omega} F & \Rightarrow K^T \models_{\alcio} F^T.
\end{align*}
}
\begin{proof}
We prove the completeness of the translation by contraposition. 
Let  $K^T \not \models_{\alcio} F^T$.  
Then there is a model 
$I=\langle \Delta, \cdot^I \rangle$ of $K^T$ in $\alcio$ such that $I$ falsifies $F$.
We show that we can build a model $J=\langle \Lambda, \cdot^J \rangle$ of $K$ in $\alc^\Omega$,
where
the domain $\Lambda$ is a transitive set
in the universe $\HF^{1/2}(\mathbb{A})$ consisting of all the hereditarily finite rational hypersets
built from atoms in $\mathbb{A}=\{{\bf a_0}, {\bf a_1},\ldots  \}$.
As a matter of fact, the domain $ \Lambda $ is to be extended possibly duplicating sets representing extensionally equal but pairwise distinct sets/elements in $\Delta$.

We define $ \Lambda $ starting from the graph\footnote{Strictly speaking the graph $ G $ introduced here is not really necessary: it is just mentioned to single out the membership relation $ \in $  from $ e^I $ more clearly.} $ G=\langle \Delta, e^I \rangle$, whose nodes are the elements of $ \Delta $ and whose arcs are the pairs $ (x,y)\in e^I $. Notice that, by Proposition \ref{finite}, the graph $ G $ can be assumed to be finite. 
Intuitively, an arc from $ x $
to $ y $ in $ G $ stands for the fact that $ y \in x $.

At this point, let $\Delta_0= \{ d_1, \ldots , d_m\} $ be the elements of $\Delta$ which, in the model $I=\langle \Delta, \cdot^I \rangle$, are not in relation $e^I$ with any other element in $\Delta$ and are non equal to the interpretation of any concept individual name $e_C$ 
(that is, $d_j \in \Delta_0 $ iff there is no $y$ such that $(d_j, y) \in e^I$ and there is no concept $C$ such that $d_j= e_C^I$).
For any given $ d\in \Delta $ we define the following hyperset $ M(d) $:
\begin{align}\label{def_M(d)}
	M(d) & = \left\{\begin{array}{ll}
						{\bf a_k} & \mbox{ if } d = d_k \in \Delta_0,\\
						\left\{M(d') \tc (d,d')\in e^I\right\} & \mbox{ otherwise. } 
					\end{array}\right.
\end{align} 
Observe that,  for the concepts $C$ occurring on the l.h.s. of membership axioms, as axiom $ C^T= \exists e^-.\{e_C\} $ is satisfied in the model $I$ of $K^T$, it holds that $d' \in (C^T)^I$ iff $(e_C^I, d') \in e^I$. Therefore, 
for $d=e_C^I$, $M(d)= M(e_C^I)= \left\{M(d') \tc (e_C^I,d')\in e^I\right\} $ $ = \left\{M(d') \tc d'\in (C^T)^I\right\}$.

The above definition uniquely determines hypersets in $ \HF^{1/2}(\mathbb{A}) $. This follows from the fact that all finite systems of (finite) set-theoretic equations have a solution in $ \HF^{1/2}(\mathbb{A}) $\footnote{More generally, when  $ e^I $ is a well-founded relation, $ M(\cdot) $ is a set-theoretic ``rendering'' of $ e^I $: the so-called \emph{Mostowski collapse} of $ e^I $ (see \cite{Jech:2003ly}).
As a consequence of the duplication of extensionally equal sets, not only we have the trivial property that, for $d,d' \in \Delta$, $d=d'$ implies $M(d)=M(d')$, but also the converse implication, i.e., 
$M(d)=M(d')$  implies $d=d'$.
}. 

Our task now is to complete the definition of $ J= \langle \Lambda, \cdot^J\rangle $ in such a way to prove that $ J $ is a model of $ K $ in $\alc^\Omega$ falsifying $ F $.
The definition is completed as follows:

	-  $\Lambda=\{M(d) \tc d \in \Delta\}$;    
	
	-  for all $B \in N_C$, $B^{J}=\{M(d) \tc  d \in B^I\};$
	
	-  for all roles $R \in N_R$ such that $ R \neq e $, 
		$R^{J}= \{(M(d),M(d')) \tc (d,d')\in R^I\};$
	
	-  
	          for all standard  named individuals $a \in N_I$ such that $a^I= d_k$, let $a^J= M(d_k) ={\bf a_k}\in \mathbb{A}$.\\
By construction,  $\Lambda$ is transitive set in a model $\emme$ of $\Omega$ 
(in fact, for all $M(d) \in \Lambda$, if $M(d') \in M(d)$, then $(d', d) \in e^I$ and then $d' \in \Delta$; therefore, $M(d') \in \Lambda$).
We can now prove, by induction on the structural complexity of concepts, that the following holds, for all $x \in \Delta$: 
\begin{align}\label{set-eq}
	M(x) \in C^J & \mbox{ if and only if } x \in (C^T)^I.
\end{align}

The base case for concept names, $ \top $, and $ \bot $ is trivial, as $\top^T= \top$, and $\bot^T = \bot$.

For the case $C=B \in N_C$,  by definition of $J$, $M(x) \in B^J $ iff $x \in B^I$.
As $B^T=B$, $M(x) \in B^J $ iff  $x \in (B^T)^I$. 

The inductive step in case $ C = C_1 \sqcap C_2 $ follows directly from the inductive hypothesis. 
If $M(x) \in (C_1 \sqcap C_2)^J$, then  $M(x) \in C_1^J$ and  $M(x) \in  C_2^J$. By inductive hypothesis, $x \in (C_1^T)^I$ and $x \in (C_2^T)^I$.
Hence, $x \in ((C_1 \sqcap C_2)^T)^I$. The vice-versa is proved similarly. 

The cases in which $ C=(\exists R.D) $ or $ C=(\forall R.D) $, are also straightforward. 
We only consider the case   $ C=(\exists R.D) $.
If $M(x) \in (\exists R.D)^J$, then there is a $M(d) \in \Lambda$ such that: $(M(x),M(d)) \in R^J$ and $M(d) \in D^J$.
By inductive hypothesis, $d \in (D^T)^I$ and, by definition of $J$,  $(x,d) \in R^I$. Hence,  $x \in ((\exists R.D)^T)^I$.
The vice-versa is proved similarly. 

For the case $ C=\texttt{Pow}(D)  $, by definition of translation, we have that: 
\begin{align*}
(C^T)^I & = ((\texttt{Pow}(D))^T)^I = (\forall e.D^T)^I  =\{x \in \Delta \tc \forall y( (x,y) \in e^I  \imp  y \in (D^T)^I\}
\end{align*}
and 
$C^J  =(\texttt{Pow}(D))^J=\mathit{Pow}(D^J) \cap \Lambda$.

Consider, for $ x \in \Delta $, $ M(x)\in \mathit{Pow}(D^J) \cap \Lambda $, which is as to say that $ M(x)\subseteq D^J$. 
All the elements of $ M(x) $ are of the form $ M(y) $ for some $ y\in \Delta $, therefore we have that: 
\begin{align*}
\forall M(y)(M(y)\in M(x) \rightarrow M(y) \in D^J), 
\end{align*}
which, by definition of $ M(\cdot) $ and by inductive hypothesis, means that: 
\begin{align*}
\forall y((x,y) \in e^I \rightarrow y \in (D^T)^I), 
\end{align*}
which means $x \in (\forall e.D^T)^I =((\texttt{Pow}(D))^T)^I$
and proves (\ref{set-eq}) in this case.


We can now use (\ref{set-eq}) to prove that axioms and assertions in $ K $ are satisfied in $ J $.


The cases $ C \sqsubseteq D $ and $ D(a) $, with $ C,D$ concepts of $\alc^\Omega$ and $ a \in N_I $, follow directly   
from (\ref{set-eq}), from the definition of $ M(\cdot) $ and
from the fact that $ C^T \sqsubseteq D^T $ and  $ D^T(a) $ (respectively) are satisfied in the model $I$ of $K^T$.

For each membership axiom $ C \in D $ in $K$, we have to show that $ C^J  \in D^J $.    
As the assertion $ D^T(e_C) $ is in $K^T$ and is satisfied in $I$, we have $e_C^I \in (D^T)^I$.
Hence, from (\ref{set-eq}),  $M(e_C^I) \in D^J$. As we have seen above, 
$M(e_C^I) = \left\{M(d') \tc d'\in (C^T)^I\right\}$  and, again from (\ref{set-eq}), $M(e_C^I) = C^J$.
Thus $ C^J \in D^J $. 

For each role membership axiom $ (C,D) \in R $ in $K$, we  show that $ (C^J , D^J) \in R^J $.
As the assertion $ R(e_C,e_D) $ is in $K^T$ and is satisfied in $I$, we have $(e_C^I,e_D^I) \in R^I$.
Hence, from the definition of $R^J$,  $(M(e_C^I), M(e_D^I)) \in R^J$. As we have seen above,
$M(e_C^I) = \left\{M(d') \tc d'\in (C^T)^I\right\}$ and, from (\ref{set-eq}), $M(e_C^I) = C^J$.
Similarly, $M(e_D^I) = D^J$.
Thus $ (C^J , D^J  ) \in R^J $. 
\hfill $\Box$
\end{proof}

\bigskip


\noindent
{\bf Proposition \ref{set-th-translation} (Soundness and Completeness of the translation of $\al^\Omega$)} 

{\em
For all concepts  $C$ and $D$ on the language of the knowledge base $K$: 
\begin{align*}
&K \models_{\al^\Omega} C \sqsubseteq D  \text{ if and only if } \\
&\Omega \models  \forall x  (  Trans(x)  
\rightarrow \forall x_1, \ldots, \forall x_n ( \bigwedge \mathit{ABox}_{\cal A} \wedge   \bigwedge \mathit{TBox}_{\cal T}  \rightarrow C^S \cap x \subseteq D^S))
\end{align*} 
where $Trans(x)$ stands for $\forall y (y \in x  \rightarrow y \subseteq x)$, that is, $x \subseteq \texttt{Pow}(x)$.
}

\begin{proof}

($\Rightarrow$) For the completeness, we proceed by contraposition. Suppose there is a model $\emme$ of $\Omega$, with universe $\U$ over $\mathbb{A}$, which falsifies the formula:

$\forall x  (  Trans(x) 
\rightarrow \forall x_1, \ldots, \forall x_n ( \bigwedge \mathit{ABox}_{\cal A} \wedge   \bigwedge \mathit{TBox}_{\cal T}  \rightarrow C^S \cap x \subseteq D^S))$\\
Then there must be some $u \in \U$, such that 
$Trans(x)$ $[u/x]$ is satisfied in $\emme$, while 
$( \forall x_1, \ldots,$ $ \forall x_n ( \bigwedge \mathit{ABox}_{\cal A} \wedge   \bigwedge \mathit{TBox}_{\cal T}  \rightarrow C^S \cap x \subseteq D^S)))[u/x]$ is falsified in $\emme$.

Hence, there must be  $v_1, \ldots, v_n$ in $\U$, such that 
$( \bigwedge \mathit{ABox}_{\cal A} \wedge   \bigwedge \mathit{TBox}_{\cal T} )[u/x,  \overline{v}/ \overline{x}]$
is satisfied in $\emme$, while $ (C^S \cap x \subseteq D^S) [u/x, \overline{v}/ \overline{x}]$ is falsified in $\emme$.
Let $\beta=[u/x, \overline{v}/ \overline{x}]$.

We define an $\al^\Omega$ interpretation $I=(\Delta, \cdot^I)$, as follows:

- $\Delta=u$;

- $A_i^I=  v_i \cap u$, for all $i=1, \ldots, n$ such that $A_i$ occurs in $K$; $A^I=  \emptyset$ for all other $A \in N_C$. 

\noindent
$I$ is well-defined.
By construction, $\Delta$ is a transitive set  living in the universe $\U$ of the $\Omega$ model $\emme$,
and $A_i^I \subseteq \Delta$.

We can prove by structural induction that, for all the concepts $C$ built from the concept names in $K$,
for the variable substitution $\beta= [u/x, \overline{v}/ \overline{x}]$, and
for all $w \in \Delta$:
\begin{align}\label{correspondence_C_C*_al}
w \in C^I  \mbox{ if and only if }  w \in (C^S)^\emme_\beta  
\end{align}
The proof is by induction on the structure of the concept $C$. 
We consider the two interesting cases of named concepts and the power-set concept.
Note that, $C^I=\{ w \mid \; w \in \Delta \mbox{ and }w \in (C^S)^\emme_\beta \}$.

\medskip
\noindent
Let $C=A_i$, for some $A_i \in N_C$ occurring in $K$.

$w  \in A_i ^I$  {\em iff}  $w  \in v_i \cap u$, with $v_i= (x_i)^\emme_\beta$ (by definition of $ A_i ^I$) 


$\;$ \ \ \ \ \ \ \ \ \ \ \ \ \  {\em iff} $w \in (x_i)^\emme_\beta$ (as $w \in  u =\Delta$)

$\;$ \ \ \ \ \ \ \ \ \ \ \ \ \  {\em iff} $w \in (A_i^S)^\emme_\beta$ 
(by the translation for named concepts)

\medskip
\noindent
Let $C=\texttt{Pow}(D)$.
By inductive hypothesis: $D^I = (D^S)^\emme_\beta$ 

$w  \in \texttt{Pow}(D) ^I$  {\em iff}  $w  \in Pow(D^I) \cap \Delta$, by the semantics of $\al^\Omega$

$\;$ \ \ \ \ \ \ \ \ \ \ \ \ \    {\em iff}  $w  \subseteq D^I$ \ \  and $w \subseteq \Delta$ \ \ (by transitivity of $\Delta$)


$\;$ \ \ \ \ \ \ \ \ \ \ \ \ \    {\em iff}  $w \subseteq  (D^S)^\emme_\beta $ and $w \subseteq \Delta$ \ \ (by inductive hypothesis) 

$\;$ \ \ \ \ \ \ \ \ \ \ \ \ \    {\em iff}  $w \subseteq  (D^S)^\emme_\beta $ \ \ (by transitivity of $\Delta$, as $w \in \Delta$)

$\;$ \ \ \ \ \ \ \ \ \ \ \ \ \    {\em iff}  $w \in  ({Pow}(D^S))^\emme_\beta$

$\;$ \ \ \ \ \ \ \ \ \ \ \ \ \    {\em iff}  $w \in  ((\texttt{Pow}(D))^S)^\emme_\beta$ 
\normalcolor
\medskip

The equivalence (\ref{correspondence_C_C*_al}) can be used to prove that the $\alc^\Omega$ interpretation $I$ is a model of $K$, 
i.e. it satisfies all axiom inclusions and membership inclusions in $K$, and that $I$  falsifies the inclusion $C \sqsubseteq D$. 
From this, it follows that, $K \not \models_{\alc^\Omega} C \sqsubseteq D$.
We prove that $I$ is a model of $K$.

For inclusion axioms, let $C \sqsubseteq D$ in $K$, we show that, 
for all $w \in \Delta$, if $w \in C^I$ then $w \in D^I$.
The inclusion $C^S \cap x \subseteq D^S$ is in $\mathit{TBox}_{\cal A}$.
We know that $( \bigwedge \mathit{TBox}_{\cal A} )^\emme_\beta$
is satisfied in $\emme$.
Hence, $(C^S)^\emme_\beta \cap u \subseteq (D^S)^\emme_\beta$ holds in $\emme$.
Suppose that $w \in C^I$. From (\ref{correspondence_C_C*_al}),  $w \in (C^S)^\emme_\beta$.
As $w \in \Delta=u $, $w \in (C^S)^\emme_\beta \cap u$. 
Therefore, $w \in (D^S)^\emme_\beta$.
Again from (\ref{correspondence_C_C*_al}), $w \in D^I$.

For membership axioms, let $C \in D$ in $K$.
We want to show that $C^I \in D^I$.
We know that  $C^S \in D^S \cap x$ is in $\mathit{ABox}_{\cal A}$
and that $( \bigwedge \mathit{ABox}_{\cal A} )^\emme_\beta$
is satisfied in $\emme$.
Hence, $(C^S)^\emme_\beta \in (D^S)^\emme_\beta \cap u$ holds in $\emme$.
Thus $(C^S)^\emme_\beta \in (D^S)^\emme_\beta$ and $(C^S)^\emme_\beta \in u$.
As $(C^S)^\emme_\beta \in \Delta$, by (\ref{correspondence_C_C*_al}), we get $(C^S)^\emme_\beta \in D^I$.
To show that $(C^S)^\emme_\beta= C^I$, from which $C^I \in D^I$ follows,
observe that: $(C^S)^\emme_\beta \in \Delta$ and  by transitivity of $\Delta$, $(C^S)^\emme_\beta \subseteq u$.
From (\ref{correspondence_C_C*_al}),
$C^I = (C^S)^\emme_\beta \cap u $, and thus $C^I = (C^S)^\emme_\beta$.

\medskip

($\Leftarrow$) For the soundness of the translation, we proceed, again, by contraposition.
Let  $I=(\Delta, \cdot^I)$ be $\al^\Omega$ model of $K$, falsifying the inclusion $C \sqsubseteq D$.
By construction, $\Delta$ is a transitive set living in the universe $\U$ of an $\Omega$ model $\emme$.

We show that $\emme$ falsifies the formula:
\begin{align} \label{formula_set_th_4.2}
\forall x  (  Trans(x) \rightarrow \forall x_1, \ldots, \forall x_n ( \bigwedge \mathit{ABox}_{\cal A} \wedge   \bigwedge \mathit{TBox}_{\cal T}  \rightarrow C^S \cap x \subseteq D^S))
\end{align}
Let $\beta$ be the variable substitution $\beta= [u/x,  \overline{v}/ \overline{x}]$,
where: $u= \Delta$ and  $v_j= A_j^I $,  for all $j=1,\ldots,n$.

We can prove that, for all the concepts $C$ built from the concept names in $K$,
and for all $d \in \Delta$:
\begin{align}\label{correspondence_C_C*_al_4.2}
d \in C^I  \mbox{ if and only if }  d \in (C^S)^\emme_\beta  
\end{align}
The proof is by induction on the structure of the concept $C$. 
Let $d \in \Delta$. We consider the two cases of named concepts and the power-set concept.

\medskip
\noindent
Let $C=A_i$, for some $A_i \in N_C$ occurring in $K$. 

$d  \in A_i ^I$  {\em iff}  $d \in v_i$,  by definition of $v_i$ 


$\;$ \ \ \ \ \ \ \ \ \ \ \ \ \  {\em iff} $d \in (x_i)^\emme_\beta$ 

$\;$ \ \ \ \ \ \ \ \ \ \ \ \ \  {\em iff} $d \in (A_i^S)^\emme_\beta$ 
(by the translation for named concepts).

\medskip
\noindent
Let $C=\texttt{Pow}(D)$.

$d  \in \texttt{Pow}(D) ^I$  {\em iff}  $d  \in Pow(D^I) \cap \Delta$, by the semantics of $\al^\Omega$

$\;$ \ \ \ \ \ \ \ \ \ \ \ \ \    {\em iff}  $d  \in Pow(D^I) $,  \ as $d \in \Delta$

$\;$ \ \ \ \ \ \ \ \ \ \ \ \ \    {\em iff}  $d  \subseteq D^I$ and $d \subseteq \Delta$, \ \ by transitivity of $\Delta$

$\;$ \ \ \ \ \ \ \ \ \ \ \ \ \    {\em iff}  $d \subseteq  (D^S)^\emme_\beta $ and $d \subseteq \Delta$, \ \ by inductive hypothesis

$\;$ \ \ \ \ \ \ \ \ \ \ \ \ \    {\em iff}  $d \in  ({Pow}(D^S))^\emme_\beta$, \ \ by transitivity of $\Delta$, as $d \in \Delta$

$\;$ \ \ \ \ \ \ \ \ \ \ \ \ \    {\em iff}  $d \in  ((\texttt{Pow}(D))^S)^\emme_\beta$, by the translation of the power-set.
\normalcolor

\noindent
Property (\ref{correspondence_C_C*_al_4.2}) can be used to prove that  the formula (\ref{formula_set_th_4.2}) is falsified in the model $\emme$ of $\Omega$.
 It is enough to prove that: 
$ ( \bigwedge \mathit{ABox}_{\cal A} \wedge   \bigwedge \mathit{TBox}_{\cal T})\beta$
is satisfied in $\emme$ and that 
$(C^S \cap x \subseteq D^S) \beta$ is falsified in $\emme$.

To prove that $(\bigwedge \mathit{TBox}_{\cal T})\beta$ holds in $\emme$,
let the inclusion $C^S \cap x \subseteq D^S$ be in $\mathit{TBox}_{\cal T}$.
Then,  $C\sqsubseteq D$ is in $K$, and is satisfied in $I$.
To show that, $(C^S)^\emme_\beta \cap u \subseteq (D^S)^\emme_\beta$ holds in $\emme$,
let $d \in (C^S)^\emme_\beta \cap u$.
By  (\ref{correspondence_C_C*_al_4.2}), $d \in C^I$. Then, $d \in D^I$ and,
again by  (\ref{correspondence_C_C*_al_4.2}), $d \in (D^S)^\emme_\beta$.

To prove that $ ( \bigwedge \mathit{ABox}_{\cal A})\beta$ holds in $\emme$,
let the inclusion $C^S \in D^S  \cap x$ be in $\mathit{ABox}_{\cal A}$.
Then,  $C\in D$ is in $K$, and is satisfied in $I$,
i.e., $C^I \in D^I$. As $D^I \subseteq \Delta$, $C^I \in \Delta$.
Let $d= C^I$. By (\ref{correspondence_C_C*_al_4.2}), $d \in (D^S)^\emme_\beta$ and,
as $u= \Delta$, $d \in (D^S)^\emme_\beta \cap u =(D^S \cap x)^\emme_\beta$.
Again by (\ref{correspondence_C_C*_al_4.2}), 
$(C^S)^\emme_\beta= C^I$.
Thus, $(C^S)^\emme_\beta \in (D^S \cap x)^\emme_\beta$.

In a similar way we can show that the inclusion $(C^S \cap x \subseteq D^S)\beta$ is falsified in $\emme$.
Indeed, $C \sqsubseteq D$ is falsified in $I$, i.e., for some $d \in \Delta$, $d \in C^I$ and $d \not \in  D^I$.
Clearly, $d \in u$ and, by (\ref{correspondence_C_C*_al_4.2}), $d \in (C^S)^\emme_\beta$.
Hence, $d \in (C^S \cap x)^\emme_\beta$. 
As $d \not \in  D^I$,  $d \not \in (D^S)^\emme_\beta$.
\hfill $\Box$
\end{proof}

\bigskip

\noindent
{\bf Proposition \ref{Prop:encoding} (Soundness and Completeness of the encoding of $\alc^\Omega$ in $\al^\Omega$)} 
{\em
 \begin{center}
$K \models_{\alc^\Omega} F  \mbox{ if and only if } K^E \models_{\al^\Omega} F^E$
\end{center}
}

\begin{proof}
($\Leftarrow$) The soundness is proved by contraposition.
Assume that $K \not \models_{\alc^\Omega} F $, then, there is a model $I=(\Delta, \cdot^I)$ of $K$ such that 
$F$ is falsified in $I$.

For the finite model property of $\alc^\Omega$, we can assume without loss of generality that the model $I$ is finite.
To build from $I$ a finite $\al^\Omega$ model $J=(\Delta', \cdot^J)$ of $K^E$ which falsifies $F^E$, 
we define a graph $G=(N , E)$ where: 
$N= \Delta \cup D_1 \cup \ldots \cup D_k$ 
and  $D_i = \{u^i_{s,t} \mid s,t \in \Delta \wedge (s,t) \in R_i^I \}$.
$E$ is defined as follows: 
\begin{center}
$E =  \bigcup_{i=1}^k \{(s, u^i_{s,t}), (u^i_{s,t},t) \mid s,t \in \Delta \wedge (s,t) \in R_i^I\}  \cup \{(s, t) \mid s,t \in \Delta \wedge t \in s\}$
\end{center}
We define an injection $\pi$ from the leaves of $N$, i.e. nodes without any successor, to  $\mathbb{A}$
and, for any given $ d\in N $, we define the following hyperset $ M(d) $: 
\begin{align*}
	M(d) & = \left\{\begin{array}{ll}
						\pi(d) & \mbox{ if } d  \mbox{ is a leaf of } N,\\
						\left\{M(d') \tc (d,d') \in E \right\} & \mbox{ otherwise. } 
					\end{array}\right.
\end{align*} 
The above definition uniquely determines hypersets in $ \HF^{1/2}(\mathbb{A}) $. This follows from the fact that all finite systems of (finite) set-theoretic equations have a solution in $ \HF^{1/2}(\mathbb{A}) $. 

Let $\Lambda=\{ M(d) \mid d \in N\}$,  possibly extended  by duplicating M(d)'s to represent extensionally-equal (bisimilar) sets corresponding to pairwise distinct elements in $N$.
As a consequence, as in previous cases,  for $d,d' \in N$, $d=d'$ if and only if $M(d)=M(d')$, i.e., there are distinct sets in $\Lambda$ for pairwise distinct elements of $N$.

Observe that, by definition of $\Lambda$,  if $(s,t) \in R_i^I$, for $s,t \in \Delta$, then 
there is some $d \in D_i$, such that $M(d) \in M(s)$ and $M(t) \in M(d)$ (and, in particular, $d= u_{s,t}^i$); and vice-versa.

Our task now is to complete the definition of $ J= \langle \Delta', \cdot^J\rangle $ in such a way to prove that $ J $ is a model of $ K^E $ in $\al^\Omega$ falsifying $ F^E$.
The definition is completed as follows:

	-  $\Delta'=\Lambda $;   
	  
	-  $A^{J}=\{ M(d) \mid d \in A^I\} $  for all $A \in N_C$,  in the language of $\alc^\Omega$;

	-  $B_i^J = M(a_i^I)= \pi(a_i^I)$, $i=1, \ldots, r$;
	         
	- $U_i^J= \{ M(u_{s,t}^i) \mid s,t \in \Delta \mbox{ and } (s,t) \in R_i^I\}$; 

	- $(F_{h,j}^i)^J = M(u_{a_h^I,a_j^I}^i) $;
	
	- $(G_{C_h,C_j}^i)^J = M(u_{s,t}^i)$, for  $s=C_h^I$ and $t= C_j^I$.

\noindent
By construction,  $\Delta'$ is transitive set in a model $\emme$ of $\Omega$.
Notice that  $B_i^J = M(a_i^I) \in \mathbb{A}$, and hence $B_i^J$ has no elements.
Notice also that, in the definition of $(G_{C_h,C_j}^i)^J$, $s$ and $t$ are elements of $\Delta$ and $(s,t) \in R_i^I$, so that $u_{s,t}^i \in E$.
In fact, $s=C_h^I$ and $t= C_j^I$ and $R_i(C_h, C_j)$ is in ${\cal A}$.
Therefore, as $I$ satisfies  the ABox ${\cal A}$, $(C_h^I,C_j^I) \in R_i^I \subseteq \Delta \times \Delta$,
and $C_h^I,C_j^I \in \Delta$.
In the following, we let $M(D_i) = \{ M(u_{s,t}^i) \mid s,t \in \Delta \mbox{ and } (s,t) \in R_i^I\} $.
%
%

It can be shown that, for all $d\in \Delta$, for all concepts $C$ in $K$ (or $F$),
\begin{align}\label{encodingE_soundness}
d \in C^I \mbox{ if and only if } M(d) \in (C^E)^J, 
\end{align}
We prove ( \ref{encodingE_soundness}) by induction on the structural complexity of concepts.
Let $d \in \Delta$. We consider the cases of named concepts and the power-set concept.

\medskip
\noindent
Let $C'=A_i$, for some $A_i \in N_C$.

$M(d)  \in (A_i^E)^J$  {\em iff}  $M(d) \in A_i^J$ 
(by the encoding for named concepts,  $A_i^E= A_i$)

$\;$ \ \ \ \ \ \ \ \ \ \ \ \ \  {\em iff} 
$d  \in A_i^I$ (by definition of $ A_i ^J$) 

\medskip
\noindent
Let $C'=\texttt{Pow}(C)$.

$M(d)  \in ((\texttt{Pow}(C))^E) ^J$  {\em iff} 

$\;$ \ \ \ \    {\em iff} $M(d)  \in ( \texttt{Pow}(  U_1 \sqcup \ldots \sqcup U_k   \sqcup  C^E)) ^J$ (by the encoding $E$)

$\;$ \ \ \ \    {\em iff} $M(d)  \in {Pow}((  U_1 \sqcup \ldots \sqcup U_k   \sqcup  C^E) ^J) \cap \Delta'$ (semantics of $\al^\Omega$)

$\;$ \ \ \ \     {\em iff}  $M(d)  \subseteq (  U_1 \sqcup \ldots \sqcup U_k   \sqcup  C^E) ^J$ and $M(d) \in \Delta'$ \footnote{
We omit condition $M(d) \in \Delta'$ in the subsequent equivalences, as it holds from the hypothesis that $d \in \Delta$}

$\;$ \ \ \ \   {\em iff}  $M(d)  \subseteq   U_1^J \cup \ldots \cup U_k^J   \cup (C^E) ^J$ 

$\;$ \ \ \ \   {\em iff}  $M(d)  \subseteq   M(D_1) \cup \ldots \cup M(D_k)   \cup (C^E) ^J$

$\;$ \ \ \ \ {\em iff}  $ \forall M(d') \in M(d)$, $M(d')  \in M(D_i)$ for  $i \in \{1,\ldots, k\}$, 
or  $M(d') \in  (C^E) ^J$ 


$\;$ \ \ \ \   {\em iff}   $ \forall d' \in d$, $d'= u_{st}^i$ 
 for some $i$ and $s,t \in \Delta$, or  $M(d') \in  (C^E) ^J$  

$\;$ \ \ \ \ \ \ \ \ (definition of $D_i$)

$\;$ \ \ \ \   {\em iff}   $ \forall d' \in d$, if $d' \neq u_{st}^i$, for all  $i \in \{1,\ldots, k\}$, then   $M(d') \in  (C^E) ^J$  

$\;$ \ \ \ \   {\em iff}   $ \forall d' \in d$, if $d' \in \Delta$,  then  $M(d') \in  (C^E) ^J$  

$\;$ \ \ \ \   {\em iff}   $ \forall d' \in d$, if $d' \in \Delta$,  then  $d'\in  C^I$  (by inductive hypothesis)

$\;$ \ \ \ \   {\em iff}   $ \forall d' \in d$,   $d'\in  C^I$   (as $d \in \Delta$ and $\Delta$ is transitive, $d' \in \Delta$)

$\;$ \ \ \ \   {\em iff}   $d \subseteq  C^I$  

$\;$ \ \ \ \   {\em iff}   $d \in  Pow(C^I)$  

$\;$ \ \ \ \   {\em iff}   $d \in  Pow(C^I) \cap \Delta$  (as $d \in \Delta$)

$\;$ \ \ \ \   {\em iff}   $d \in ( \texttt{Pow}(C))^I$  (semantics of $\alc^\Omega$)

\medskip

\noindent
Let $C'= \forall R_i. C$ 

$M(d)  \in ((\forall R_i. C)^E) ^J$  {\em iff} 

$\;$ \ \ \ \    {\em iff} $M(d)  \in ( \texttt{Pow}(\neg  U_i  \sqcup  \texttt{Pow}(C^E))) ^J$ (by the encoding $E$)

$\;$ \ \ \ \    {\em iff} $M(d)  \in {Pow}((\neg  U_i  \sqcup  \texttt{Pow}(C^E)) ^J) \cap \Delta'$ (semantics of $\al^\Omega$)

$\;$ \ \ \ \     {\em iff}  $M(d)  \subseteq ( \neg  U_i \sqcup   \texttt{Pow}(C^E) ) ^J$ and $M(d) \in \Delta'$


$\;$ \ \ \ \     {\em iff}  $M(d)  \subseteq  (\neg  U_i)^J \cup  (\texttt{Pow}(C^E))^J$ \ \  \  ($M(d) \in \Delta'$ is omitted as it holds from  $d \in \Delta$)

$\;$ \ \ \ \     {\em iff}  $M(d)  \subseteq  (\Delta' \backslash  U_i^J) \cup   (\texttt{Pow}(C^E))^J$ 

$\;$ \ \ \ \     {\em iff}  $M(d)  \subseteq (\Delta' \backslash  M(D_i)) \cup   (\texttt{Pow}(C^E))^J$ 

$\;$ \ \ \ \ {\em iff}  $ \forall M(d') \in M(d)$, $M(d')  \not \in M(D_i)$ 
or  $M(d') \in   {Pow}((C^E )^J) \cap \Delta'$ 

$\;$ \ \ \ \ {\em iff}  $ \forall M(d') \in M(d)$, if $M(d')  \in M(D_i)$, then  $M(d') \subseteq   (C^E )^J \cap \Delta'$

$\;$ \ \ \ \ {\em iff}  $ \forall M(d') \in M(d)$, if $M(d')  \in M(D_i)$, then  $M(d') \subseteq   (C^E )^J $

$\;$ \ \ \ \ \ \ \ \ \ \ \ \ \ \ \ \ \   (by transitivity of $\Delta'$, $M(d') \in \Delta'$)

$\;$ \ \ \ \ {\em iff}  $ \forall M(d') \in M(d)$, if $M(d')  \in M(D_i)$, then    
$\forall M(d'') \in M(d')$,  $M(d'') \in (C^E )^J$ 


$\;$ \ \ \ \ {\em iff}  $ \forall d' \in N$, if $M(d') \in M(d)$ and $M(d')  = M(u_{s,t}^i)$, for some $u_{s,t}^i \in D_i$,

$\;$ \ \ \ \ \ \ \ \ \ \ \ \ \ \ \ \ \  \ \ \ \ \ \ \ \ \ \ \ \ \ \ \ \ \ \ \ \ \ \  \ \ \ \ \ \ \ \ \ \ \  \ \ \ \ \ \ \ \ \ \ \ \ \ \ \ \ \ \ \ \ \ \ then  $\forall M(d'') \in M(d')$,  $M(d'') \in (C^E )^J$ 

$\;$ \ \ \ \ \ \ \ \ \ \ \ \ \ \ \ \ \  (by definition of $M(D_i) $)


%

$\;$ \ \ \ \ {\em iff}  $ \forall d' \in N$, if  $M(d')  = M(u_{s,t}^i)$,  for some $u_{s,t}^i \in D_i$,  
and $(d,d')\in E$,

$\;$ \ \ \ \ \ \ \ \ \ \ \ \ \ \ \ \ \  \ \ \ \ \ \ \ \ \ \ \ \ \ \ \ \ \ \ \ \ \ \  \ \ \ \ \ \ \ \ \ \ \  \ \ \ \ \ \ \ \ \ \ \ \ \ \ \ \ \ \ \ \ \ \ then  $\forall M(d'') \in M(d')$,  $M(d'') \in (C^E )^J$

$\;$ \ \ \ \ \ \ \ \ \ \ \ \ \ \ \ \ \   (by propositional reasoning and the definition of hypersets in $\Lambda$ from graph $G$)

$\;$ \ \ \ \   {\em iff}   $ \forall d' =u_{s,t}^i \in D_i$ if $(d,u_{s,t}^i)\in E$,  then   
$\forall M(d'') \in M(d')$,  $M(d'') \in (C^E )^J$ 


$\;$ \ \ \ \ \ \ \ \ \ \ \ \ \ \ \ \ \ (as there are distinct sets in $\Lambda$ for pairwise distinct elements of $N$) 

\normalcolor

$\;$ \ \ \ \   {\em iff}   $ \forall u_{s,t}^i \in D_i$ if $(d,u_{s,t}^i)\in E$,  then   
$\forall d'' \in N$, such that  $(u_{s,t}^i,d'')\in E$, $M(d'') \in (C^E )^J$ 

$\;$ \ \ \ \   {\em iff}   $ \forall u_{d,t}^i \in D_i$ if $(d,u_{d,t}^i)\in E$,  and $(u_{d,t}^i,t)\in E$, then $M(t) \in (C^E )^J$ 

$\;$ \ \ \ \ \ \ \ \ \ \ \ \ \ \ \ \ \    (by definition of $E$, $s=d$ and $d''=t$)

$\;$ \ \ \ \   {\em iff}   $ \forall u_{d,t}^i \in D_i$ if $(d,u_{d,t}^i)\in E$ and $(u_{d,t}^i,t)\in E$, then   $t \in C^I$  
(inductive hypothesis)

$\;$ \ \ \ \   {\em iff}   $ \forall t \in \Delta$ if $(d,t) \in R^I_i$, then   $t \in C^I$  
(definition of $E$)

$\;$ \ \ \ \   {\em iff}   $d \in  (\forall  R_i. C)^I$  (semantics of $\alc^\Omega$)

\medskip

\noindent
Equivalence (\ref{encodingE_soundness}) can be used to prove that $J$ is a model of $K^E$ that falsifies $F^E$,
thus showing that $K^E \not \models_{ \al^\Omega} F^E$.

Let us prove that $J$ is a model of $K^E$. We consider the interesting cases.

For $C^E \sqcap \neg (U_1 \sqcup \ldots \sqcup U_k) \sqsubseteq D^E$ in $K^E$, the inclusion axiom $C  \sqsubseteq D$ is in $K$,
and is satisfied in $I$,
that is, for all $d \in \Delta$, if  $d \in C^I$ then $d \in D^I$.
For $M(d) \in \Delta'$, let $M(d) \in (C^E \sqcap \neg (U_1 \sqcup \ldots \sqcup U_k))^J$. 
Then, $M(d) \in (C^E)^J$ and   $M(d) \not \in (U_1)^J \cup \ldots \cup (U_k)^J$, 
i.e. $M(d) \not \in M(D_j)$, for all $j= 1, \ldots, k$. Hence, $d \in \Delta$. 
By (\ref{encodingE_soundness}),
$d \in C^I$, and then  $d \in D^I$.
Again by  (\ref{encodingE_soundness}), $M(d) \in (D^E)^J$.

For the membership axioms, let $C^E \in D^E$ be in $K^E$.
The membership axiom  $C \in D$ in $K$ and is satisfied in $I$,
i.e. $C^I \in D^I$. As $D^I \subseteq \Delta$, $C^I \in \Delta$ and,
by (\ref{encodingE_soundness}),
$M(C^I) \in (D^E)^J$.
Again by (\ref{encodingE_soundness}),
$M(C^I)= \{ M(d) \mid d \in C^I\} = (C^E)^J$,
Hence, $(C^E)^J \in (D^E)^J$.

For each assertion $R_i(C_h,C_j)$  in $K$, we have to show that
the membership axioms $G_{C_h,C_j}^i \in C_h^E$,  $C_j^E \in G_{C_h,C_j}^i$ and $G_{C_h,C_j}^i \in U_i$
added to $K^E$ by encoding $R_i(C_h,C_j)$ 
are satisfied in $J$, that is:
  $(G_{C_h,C_j}^i)^J \in (C_h^E)^J$,  $(C_j^E)^J \in (G_{C_h,C_j}^i)^J$ and $(G_{C_h,C_j}^i)^J \in U_i^J$.
As assertion $R_i(C_h,C_j)$  is satisfied in $I$, that is $(C_h^I,C_j^I) \in R_i^I$.
By construction of $E$ there is $u_{s,t}^i \in D_i \subseteq N$ such that $s=C_h^I$ and $t=C_j^I$,
with $(s,u_{s,t}^i), (u_{s,t}^i ,t) \in E$. 
By definition of the model $J$,  $(G_{C_h,C_j}^i)^J = M(u_{s,t}^i)$. 
Also, $ M(u_{s,t}^i) \in M(s)$ and $ M(t) \in M(u_{s,t}^i)$ hold in $J$.
Replacing $s$ and $t$ with their definitions and $M(u_{s,t}^i)$ with $(G_{C_h,C_j}^i)^J $ we get:
$(G_{C_h,C_j}^i)^J \in M(C_h^I)$ and $M(C_j^I) \in (G_{C_h,C_j}^i)^J$.
 Finally, by construction, $u_{s,t}^i \in D_i $, and $U_i^J=M(D_i)$,
 than $M(u_{s,t}^i) \in M(D_i)= U_i^J$. Therefore, $(G_{C_h,C_j}^i)^J \in U_i^J$.

It is easy to see that the axioms  $A_i \sqsubseteq \neg (U_1 \sqcup \ldots \sqcup U_k)$,  $B_i \in \neg (U_1 \sqcup \ldots \sqcup U_k)$, 
 $C^E \in \neg (U_1 \sqcup \ldots \sqcup U_k) $, 
and axiom  
$ \neg (U_1 \sqcup \ldots \sqcup U_k) \sqsubseteq \texttt{Pow}(\neg (U_1 \sqcup \ldots \sqcup U_k)  \sqcup \texttt{Pow}(\neg (U_1 \sqcup \ldots \sqcup U_k) )) $
are all satisfied in $J$ by construction.

\medskip

In a similar way, we can prove that $F$ is falsified in $I$, considering the different cases for $F$,
and given the hypothesis that $F^E$ is falsified by $J$.

%
%

\medskip

($\Rightarrow$) We sketch the proof of completeness. The proof is by contraposition.
Assume that $K^E \not \models_{\al^\Omega} F^E $, then, there is an $\al^\Omega$ model $J=(\Delta, \cdot^J)$ of $K^E$ such that 
$F^E$ is falsified in $J$.

For the finite model property of $\al^\Omega$ (which is a fragment of $\alc^\Omega$), we can assume without loss of generality that the model $J$ is finite.
We build from $J$ an $\alc^\Omega$ model $I=(\Delta', \cdot^I)$ of $K$ which falsifies $F$, 
defining $\Delta'$ as a transitive set
in the universe $\HF^{1/2}(\mathbb{A})$ consisting of all the hereditarily finite rational hypersets
built from atoms in $\mathbb{A}=\{{\bf a_0}, {\bf a_1},\ldots  \}$.

We start from the graph $G=(N , E)$, with nodes $N= \Delta \backslash (U_1^J \cup \ldots \cup U_k^J )$, whose arcs are defined as follows:
$E=   \{ (d_1, d_2) \mid   d_1, d_2  \in N \wedge  d_2 \in d_1\}$.

$ G$ is finite. 
Observe that, for each $a_i$ in $K$, $B_i^J \in N$, by axiom $B_i \in \neg (U_1 \sqcup \ldots \sqcup U_k)$.
Similarly, for each $A_i$ in $K$, $A_i^J \subseteq N$, by axiom $A_i \sqsubseteq \neg (U_1 \sqcup \ldots \sqcup U_k)$.

We define an injection $\pi$ from the leaves of $N$ (i.e. nodes without any successor) plus the elements $B_1^I, \ldots, B_r^I\in N$ 
to  $\mathbb{A}$.
For any given $ d\in N $, we define the following hyperset $ M(d) $:
\begin{align}\label{def_M(d)}
	M(d) & = \left\{\begin{array}{ll}
						\pi(d) & \mbox{ if } d  \mbox{ is a leaf of } N \mbox{ or } d=B_j^J \mbox{ for some $j$ },\\
						\left\{M(d') \tc (d,d') \in E \right\} & \mbox{ otherwise. } 
					\end{array}\right.
\end{align} 
The above definition uniquely determines hypersets in $ \HF^{1/2}(\mathbb{A}) $. This follows from the fact that all finite systems of (finite) set-theoretic equations have a solution in $ \HF^{1/2}(\mathbb{A}) $.

 $\Delta'=\{M(d) \tc d \in N\} $, possibly extended  by duplicating M(d)'s to represent extensionally-equal (bisimilar) sets corresponding to pairwise distinct elements in $N$.
We have to complete the definition of $ I= \langle \Delta', \cdot^I\rangle $ in such a way to prove that $ I $ is a model of $ K $ in $\alc^\Omega$ falsifying $ F $. The definition is completed as follows:

	
	-  $A^{I}=\{M(d)  \tc M(d) \in \Delta'  \wedge d \in A^J\}$,  for all $A \in N_C$;
	
	-  $R_i^{I}= \{(M(d),M(d')) \tc M(d), M(d') \in \Delta' \wedge \exists u \in U_i^J ( u \in d \wedge d' \in u) \}$, 
	
	   $ \mbox{  \ } $ for all roles $R_i$ occurring in $K$;  $\mbox{ }$ $R_i^{I}=\emptyset$ for all other roles $R\in N_R$;
		
	 - $a_i^I= M(B_i^J)= \pi(B_i^J)$  for all  named individuals  $a_i$ occurring in $K$;
	 
	       $ \mbox{  \ } $   $a^{I}=M(B_1^J)$ for all other $a \in N_I$.
	         	          
\noindent
By construction,  $\Delta'$ is a transitive set in a model $\emme$ of $\Omega$.
As $A^J \subseteq \Delta \backslash (U_1^J \cup \ldots \cup U_k^J )$, $A_i^I \subseteq \Delta'$.
To complete the proof it can be shown that, for all $M(d) \in \Delta'$, and $C$ in $K$ (or in $F$):
\begin{align}\label{encodingE_completeness}
 M(d) \in C^I   \mbox{ if and only if }    d \in (C^E)^J 
\end{align}
which can be used to prove that $J$ is a model of $K$ that falsifies $F$. 

We prove ( \ref{encodingE_completeness}) by induction on the structural complexity of concepts.
Let $M(d) \in  \Delta'$. We consider the cases of named concepts, the power-set concept and the universal restriction.

\medskip
\noindent
Let $C'=A_i$, for some $A_i \in N_C$.

$M(d)  \in A_i^I$  {\em iff}  $d \in A_i^J$ 
(by definition of $ A_i ^I$) 

$\;$ \ \ \ \ \ \ \ \ \ \ \ \ \  {\em iff} 
$d  \in (A_i^E)^J$ (as $ A_i^E= A_i$) 

\medskip
\noindent
Let $C'=\texttt{Pow}(C)$.

$d \in ((\texttt{Pow}(C))^E) ^J$  {\em iff} 

$\;$ \ \ \ \    {\em iff} $d  \in ( \texttt{Pow}(  U_1 \sqcup \ldots \sqcup U_k   \sqcup  C^E)) ^J$ (by the encoding $E$)

$\;$ \ \ \ \    {\em iff} $d  \in {Pow}((  U_1 \sqcup \ldots \sqcup U_k   \sqcup  C^E) ^J) \cap \Delta$ (semantics of $\al^\Omega$)

$\;$ \ \ \ \     {\em iff}  $d  \subseteq (  U_1 \sqcup \ldots \sqcup U_k   \sqcup  C^E) ^J$ 

$\;$ \ \ \ \   {\em iff}  $d  \subseteq   U_1^J \cup \ldots \cup U_k^J   \cup (C^E) ^J$ 

$\;$ \ \ \ \ {\em iff}  $ \forall d' \in d$, $d'  \in (U_1^J \cup \ldots \cup U_k^J)$ ,  or  $d' \in  (C^E) ^J$ 

$\;$ \ \ \ \ {\em iff}  $ \forall d' \in d$, if $d'  \not \in (U_1^J \cup \ldots \cup U_k^J)$,  then  $d' \in  (C^E) ^J$ 

$\;$ \ \ \ \ {\em iff}  $ \forall d' \in \Delta$, if $(d,d') \in E$,  then  $d' \in  (C^E) ^J$ 

$\;$ \ \ \ \ {\em iff}  $ \forall d' \in \Delta$, if $(d,d') \in E$,  then  $M(d') \in  C^I$ (by inductive hypothesis)

$\;$ \ \ \ \ {\em iff}  $ \forall M(d') \in \Delta'$, if $M(d')\in M(d)$,  then  $M(d') \in  C^I$ 

$\;$ \ \ \ \ {\em iff}   $M(d) \subseteq  C^I$

$\;$ \ \ \ \ {\em iff}   $M(d) \in   Pow(C^I)$

$\;$ \ \ \ \   {\em iff}   $M(d) \in  Pow(C^I) \cap \Delta'$  (as $M(d) \in \Delta'$)

$\;$ \ \ \ \   {\em iff}   $M(d) \in ( \texttt{Pow}(C))^I$  (semantics of $\alc^\Omega$)

\medskip

\noindent
Let $C'= \forall R_i. C$ 

$d \in ((\forall R_i. C)^E) ^J$  {\em iff} 

$\;$ \ \ \ \    {\em iff} $d  \in ( \texttt{Pow}(\neg  U_i  \sqcup  \texttt{Pow}(C^E))) ^J$ (by the encoding $E$)

$\;$ \ \ \ \    {\em iff} $d  \in {Pow}((\neg  U_i  \sqcup  \texttt{Pow}(C^E)) ^J) \cap \Delta$ (semantics of $\al^\Omega$)

$\;$ \ \ \ \     {\em iff}  $d  \subseteq ( \neg  U_i \sqcup   \texttt{Pow}(C^E) ) ^J$ 

$\;$ \ \ \ \     {\em iff}  $d \subseteq ( \neg  U_i^J \cup   (\texttt{Pow}(C^E) ) ^J$ 

$\;$ \ \ \ \     {\em iff}  $d  \subseteq (\Delta \backslash  U_i^J) \cup   (\texttt{Pow}(C^E))^J$ 

$\;$ \ \ \ \ {\em iff}  $ \forall d' \in d$, $d'  \not \in U_i^J$  or  $d' \in   {Pow}((C^E )^J) \cap \Delta$ 

$\;$ \ \ \ \ {\em iff}  $ \forall d' \in d$, $d'  \not \in U_i^J$  or  $d' \in   {Pow}((C^E )^J)$ (by transitivity of $\Delta$, $d' \in \Delta$)

$\;$ \ \ \ \ {\em iff}  $ \forall d' \in d$, if $d'  \in U_i^J$, then  $d' \subseteq  (C^E )^J $

$\;$ \ \ \ \ {\em iff}  $ \forall d' \in d$, if $d'  \in U_i^J$, then    $\forall d'' \in d'$,  $d'' \in (C^E )^J$ 

$\;$ \ \ \ \ \ \ \ \ \ \ \ (and by $Trans^2(\Delta \backslash (U_1 \sqcup \ldots \sqcup U_k ))$, $d'' \in N$)

$\;$ \ \ \ \   {\em iff}   $ \forall d' \in d$, if $d' \in U_i^J$,  then   $\forall d'' \in d'$,  $M(d'') \in C^I$  (by inductive hypothesis)

$\;$ \ \ \ \   {\em iff}   $ \forall d' , d'' \in N$, if $d' \in U_i^J$ and $ d' \in d$ and  $d'' \in d'$, then $M(d'') \in C^I$ 

$\;$ \ \ \ \   {\em iff}   $ \forall M(d'') \in \Delta'$, if $(M(d),M(d'')) \in R_i^I$,  $M(d'') \in C^I$ (by definition of $R_i^I$)

$\;$ \ \ \ \   {\em iff}   $M(d) \in  (\forall  R_i. C)^I$ 

\medskip

The equivalence (\ref{encodingE_completeness}) can be used to prove that $I$ is a model of $K$ that falsifies $F$,
thus showing that $K \not \models_{ \alc^\Omega} F$.

Let us prove that $I$ is a model of $K$.

For the inclusion axioms, let $C \sqsubseteq D$ be in $K$.
Then $C^E  \sqcap \neg (U_1 \sqcup \ldots \sqcup U_k ) \sqsubseteq D^E$ is in $K^E$, and is satisfied in $J$,
that is, for all $d \in \Delta$, if  $d \in (C^E)^J$ and $d \not \in U_i^J$ (for all $j=1,k$), then $d \in (D^E)^J$.
Let $M(d) \in C^I$. 
By (\ref{encodingE_completeness}),
$d \in (C^E)^J$ and, as $d \in N$, $d \not \in U_i^J$ (for all $j=1,k$).  Hence,  $d \in (D^E)^J$.
Again by  (\ref{encodingE_completeness}), $M(d) \in D^I$.

For the membership axioms, let $C \in D$ in $K$.
The membership axioms $C^E \in D^E$ and $C^E \in \neg (U_1 \sqcup \ldots \sqcup U_k )$
 are in $K^E$ and are satisfied in $J$.
i.e., $(C^E)^J \in (D^E)^J$ and $(C^E)^J  \not \in (U_1 \sqcup \ldots \sqcup U_k )^J$.
Thus 
$(C^E)^J  \in N$ and,
by (\ref{encodingE_completeness}), $M((C^E)^J) \in D^I$ .
Again by (\ref{encodingE_completeness}),
$C^I= \{ M(d) \mid d \in (C^E)^J\} = M(C^E)^J)$,
thus $C^I \in D^I$.

For each assertion $C(a_i)$ in $K$, the membership axiom $a_I^E \in C^E$ is in $K^E$.
Therefore, $(a_i^E)^J \in (C^E)^J$.
By definition of the encoding,  $B_i^J \in (C^E)^J$ and $B_i^J \in \neg (U_1 \sqcup \ldots \sqcup U_k )^J$.
Thus, $B_i^J \in N$.
By  (\ref{encodingE_completeness}), $M(B_i^J) \in C^I$.
Hence, $a_i^I \in C^I$ (by definition of the interpretation of $a_i$ in $I$).

For each assertion $R_i(a_h,a_j)$  in $K$,
we have to show that $(a_h^I,a_j^I) \in R_i^I$.
The membership axioms $F_{h,j}^i \in B_h$,  $B_j \in F_{h,j}^i$ and $F_{h,j}^i \in U_i$ are in $K^E$,
and are satisfied in $J$.
Thus,  $(F_{h,j}^i)^J \in B_h^J$,  $B_j^J \in (F_{h,j}^i)^J$ and $(F_{h,j}^i)^J \in U_i^J$.
Let $d = (F_{h,j}^i)^J \in \Delta$.
Given that $d \in U_i^J$, from $d \in B_h^J$ and $B_j^J \in d$, by definition of $R_i^I$,
and $B_h^J$ and $B_j^J$ are in $N$ (by axioms $B_h^J \in \neg (U_1 \sqcup \ldots \sqcup U_k )^J$)
and $B_j^J \in \neg (U_1 \sqcup \ldots \sqcup U_k )^J$) )
we have $(M(B_h^J), M(B_j^J)) \in R_i^I$.
By definition of $I$, $a_i^I=M(B_i^J)$ and  $a_j^I=M(B_j^J)$, therefore:  
$(a_h^I, a_j^I)) \in R_i^I$.

For each assertion $R_i(C_h,C_j)$  in $K$,
we have to show that $(C_h^I,C_j^I) \in R_i^I$.
The membership axioms $G_{C_h,C_j}^i \in C_h^E$,  $C_j^E \in G_{C_h,C_j}^i$ and $G_{C_h,C_j}^i \in U_i$ are in $K^E$,
and are satisfied in $J$.
Thus,  $(G_{C_h,C_j}^i)^J \in (C_h^E)^J$,  $(C_j^E)^J \in (G_{C_h,C_j}^i)^J$ and $(G_{C_h,C_j}^i)^J \in U_i^J$.
Let $d = (G_{C_h,C_j}^i)^J \in \Delta$.
Given that $d \in U_i^J$, from $d \in (C_h^E)^J$ and $(C_j^E)^J \in d$.
As $C_h^E \in \neg (U_1 \sqcup \ldots \sqcup U_k )$. and $C_j^E \in \neg (U_1 \sqcup \ldots \sqcup U_k )$ are in $K^E$,
$(C_h^E)^J,(C_j^E)^J in N$ and,  by definition of $R_i^I$,
$(M((C_h^E)^J), M((C_j^E)^J)) \in R_i^I$.
By  (\ref{encodingE_completeness}),  
$C_h^I= \{ M(d) \mid d \in (C_h^E)^J\} = M(C_h^E)^J$,
and similarly $C_j^I=  M(C_j^E)^J$.
Hence,
$(C_h^I, C_j^I)) \in R_i^I$.

\medskip

In a similar way, we can prove that $F$ is falsified in $I$, considering the different cases for $F$,
and given the hypothesis that $F^E$ is falsified by $J$.
\hfill $\Box$
\end{proof}

\end{appendix}

\end{document}